\newcommand{\rd}{{\rm d}}
\newcommand{\veca}{{\mathbf a}}
\newcommand{\vecr}{{\mathbf r}}
\newcommand{\vecs}{{\mathbf s}}
\newcommand{\vecu}{{\mathbf u}}
\newcommand{\vecc}{{\mathbf c}}
\newcommand{\vece}{{\mathbf e}}
\newcommand{\vecx}{{\mathbf x}}
\newcommand{\vecX}{{\mathbf X}}
\newcommand{\vecz}{{\mathbf z}}
\newcommand{\vecw}{{\mathbf w}}
\newcommand{\vecW}{{\mathbf W}}
\newcommand{\vecD}{{\mathbf D}}
\newcommand{\EE}{{\mathbb E}}

\documentclass[11pt]{article}
\usepackage{graphicx}
\usepackage{longtable,tabu}
\usepackage{epstopdf}
\usepackage{relsize,exscale}
\usepackage{amssymb}
\usepackage{bm}
\usepackage{amsmath}
\usepackage{cancel}
\usepackage[margin=1.0in]{geometry}
\usepackage{cite}
\usepackage[shortlabels]{enumitem}
\usepackage{scalerel}
\usepackage{color}
\usepackage{enumitem}
\usepackage[titletoc]{appendix}
\usepackage[font=small,labelfont=bf]{caption}
\allowdisplaybreaks
\usepackage{makecell}
\usepackage{amsthm}
\newtheorem{theorem}{Theorem}[section]
\newtheorem{corollary}[theorem]{Corollary}
\newtheorem{lemma}[theorem]{Lemma}

\newtheorem{proposition}[theorem]{Proposition}
\usepackage{chngcntr}
\begin{document}
\title{Average entropy of Gaussian mixtures}
\author{Basheer Joudeh and Boris \v{S}kori\'{c}}
\date{}
\maketitle
\begin{abstract}
We calculate the average differential entropy of a $q$-component Gaussian mixture in $\mathbb R^n$. For simplicity, all components have covariance matrix $\sigma^2 {\mathbf 1}$, while the means $\{\mathbf{W}_i\}_{i=1}^{q}$ are i.i.d. Gaussian vectors with zero mean and covariance $s^2 {\mathbf 1}$. We obtain a series expansion in $\mu=s^2/\sigma^2$ for the average differential entropy up to order $\mathcal{O}(\mu^2)$, and we provide a recipe to calculate higher-order terms. Our result provides an analytic approximation with a quantifiable order of magnitude for the error, which is not achieved in previous literature.
\end{abstract}
\section{Introduction}
\subsection{Gaussian Mixtures}
{A Gaussian} mixture probability density on $\mathbb R^n$ is a function of the following form: \begin{equation}f(x)=\sum_{i=1}^q p_i g_{w_i,K_i}(x),\end{equation}
where $q$ is some integer, $p_i$ are probabilities, and $g_{w,K}$ stands for the Gaussian distribution
$g_{w,K}(x) = (2\pi)^{-\frac n2}(\det K)^{-\frac12} \exp [-\frac12(x-w)^T K^{-1}(x-w) ]$,
with $x,w\in\mathbb R^n$. In words, the density $f(x)$ is built as a weighted average of $q$ different $n$-dimensional Gaussian distributions. A mixture like this occurs when one stochastic process, with probability mass function $p_1,\ldots,p_q$, determines which distribution is chosen for~$x$.
{Gaussian mixtures are widely used in various areas for their simplicity, as well as their wide range of applicability by virtue of being global (smooth) function approximators. Most recently, in astrophysics, Gaussian mixture models were used to study the kinematics of dark matter \cite{ASTROREF1}, as well as predicting the spectrum of quasars \cite{ASTROREF2}. Furthermore, Gaussian mixtures are also widely used in the intersection between statistical physics and machine learning, i.e., diffusion models \cite{DIFFREF1,DIFFREF2}. Diffusion models became popular for their wide range of applications, including generative tasks and image restoration, for example, see \cite{DIFFREF3,DIFFREF4} for Gaussian mixture diffusion models. In wireless communications, Gaussian mixtures were recently utilized to estimate the channel trajectories of moving mobile terminals \cite{WIREF1}, while in cybernetic security, they played a role in defending against the Byzantine attack \cite{CYBREF1}, and earlier for authentication in wireless networks \cite{CYBREF2}. In bioinformatics, they are used as an alternative to clustering algorithms for analyzing gene expression data \cite{BIOREF1,BIOREF2}. They also find application in thermodynamics, e.g., a computationally efficient Monte Carlo method was devised in \cite{THERMREF1} to study the properties of nonadiabatic systems.
In \cite{MIXREF}, many of the above-mentioned disciplines are combined, where the authors evaluate probabilities in deep generative models.}
\subsection{Related Work}
\label{pastworksec}
Analytically computing or estimating the differential entropy of a Gaussian mixture is a difficult problem.
The differential entropy of a continuous variable $X\in\cal X$, $X$$\sim$$f$ is defined as $h(X)=-\int_{\cal X} f(x)\ln f(x) {\rm d}x$ \cite{cover}.
In the special case of a single Gaussian component ($q=1$), the integral can be computed analytically, yielding
$\frac n2\ln 2\pi e +\frac12\ln \det K$.
In \cite{q2}, the authors study the case of $q=2$, $n=1$, and $w_1=-w_2$, where a numerical approximation for this case is obtained. It is a known fact that given the second moment of a random variable, the distribution with maximum entropy is the Gaussian distribution~\cite{cover}. That is, we can easily obtain a loose upper bound for the differential entropy of a Gaussian mixture given by \cite{bound2}:
\begin{equation}
\label{boundstup}
h(X) \leq \dfrac{n}{2}\ln 2\pi e+\dfrac{1}{2} \ln \mathrm{det} (\Sigma),
\end{equation}
where $\Sigma$ is given by
\begin{equation}
\Sigma=\sum_{i=1}^{q} p_i (w_i w_i^{\mathrm{T}}+K_i)-\left(\sum_{i=1}^{q}p_i w_i\right)\left(\sum_{j=1}^{q}p_j w_j^{\mathrm{T}}\right).
\end{equation}
A sequence of maximum entropy upper bounds for the differential entropy was obtained in \cite{bound},
most notably the Laplacian upper bound, which is tighter than the Gaussian upper abound in some cases.
However, it has no closed-form expression.
Another way of approximating the differential entropy is by replacing the density inside the logarithm by
a single Gaussian $\bar f$ with covariance and mean identical to the mixture.
This leads to an exact expression in terms of the relative entropy,
\begin{equation}
h(X)=\dfrac{n}{2}\ln 2\pi e+\dfrac{1}{2} \ln \mathrm{det} (\Sigma)-D(f||\bar{f}).
\end{equation}
One can then find approximations to the relative entropy, as in \cite{relat,relat2}.
Although this method is efficient, it also does not a have closed-form expression and only provides an upper bound to $h(X)$. Other Monte Carlo sampling methods give guaranteed convergence.
However, they become computationally demanding as the number of samples grows. An approximation for $h(X)$ was obtained in \cite{approx} by performing a Taylor series expansion of the logarithm term.
In order to avoid the need to include higher-order terms, a splitting of Gaussian components belonging to the density outside the logarithm is applied.
The idea is to split components with high variance and replace them with Gaussian mixtures with components of lower variance.
This is because the higher the variance of the components the higher the number of terms needed in the Taylor expansion to achieve a certain level of accuracy.
Of course, this process is not exact and produces an error depending on how many Gaussian components are included in the split. The splitting method is not directly comparable with our work, as we do not make a distinction between variances of the different components in the mixture; instead, our expansion parameter is how spread the components of the mixture are relative to their shared variance value. However, one can obtain a basic upper bound on $h(X)$, as shown in \cite{approx}:
\begin{equation}
\label{betterbound}
h(X)\leq  \sum_{i=1}^{q} p_i \left(-\ln p_i+\dfrac{1}{2}\ln (2\pi e)^n \mathrm{det}(K_i)\right).
\end{equation}
As mentioned in \cite{approx}, the upper bound \eqref{betterbound} is significantly tighter than the Gaussian bound~\eqref{boundstup};
it is exact in the case of a single Gaussian;
and it is arbitrarily close to the real value of $h(X)$ when the support shared between the components of the mixture is negligible.
A refinement of this bound was also introduced by means of merging different clusters in the mixture.
However, we are concerned with the average over all possible mixtures,
and therefore, this family of bounds is less comparable with our result. 
In \cite{dist}, the differential entropy is estimated for mixture distributions using pair-wise distances between the components of the mixture.
It is shown that the estimator is a lower bound when the Chernoff $\alpha$-divergence is used
, and an upper bound when the Kullback--Leibler divergence is used.
From the results of \cite{dist}, it is difficult to obtain tight bounds on the average entropy, which is the focus of our paper.
\subsection{Contributions}

We develop an analytical estimation method for the {\em average} Gaussian mixture entropy problem in the special case of equal weights $\frac1q$ and equal covariance matrix~$\sigma^2 {\mathbf 1}$.
Our method postulates that the displacements $w_i$ themselves are stochastic, i.i.d. with Gaussian distribution $g_{0,s^2 {\bf 1}}$.
We compute $h(X)$ averaged over the displacements~$w_i$, {resulting in the conditional entropy $h(X|W)$. We mention that the quantity $h(X|W)$, for exactly this setting with equal weights, spherical covariance, and Gaussian displacements, plays a role in the detectability of digital watermarks (see Section~\ref{sec:startingPoint}).
}

We work in the regime $s<\sigma$. Our method uses the fraction $\mu=\frac {s^2}{\sigma^2}$ as the small parameter for a power expansion.
We show results up to and including order $\mathcal{O}(\mu^2)$. Our result is novel since most available estimators of $h(X)$ in the literature have no closed-form expressions, and usually rely on upper bounds rather than an explicit calculation of $h(X)$. In \cite{approx}, although a Taylor series approximation is employed, its accuracy is dependent on the splitting mentioned in Section \ref{pastworksec}. Furthermore, it is difficult to estimate the order of magnitude for the error in a general expansion around the means. In contrast, our set-up does not need the splitting method, and for our result, it is possible to quantify the order of magnitude for the error since our power series is in $\mu$.
\section{Deriving the Series Expansion}
\subsection{Notation}
Boldface lowercase letters denote $n$-dimensional vectors (e.g., $\mathbf{a}$);
the inner product between two $n$-dimensional vectors $\mathbf{a}$ and $\mathbf{b}$ is denoted by $\mathbf{a}\cdot \mathbf{b}$;
A hat denotes an array in $q$ dimensions, i.e., indexed with the number of the Gaussian components;
for lowercase letters,
$\hat{(\cdot)}$ denotes a $q$-dimensional vector (e.g., $\hat{a}$);
for uppercase letters, a $q\times q$ matrix (e.g.,~$\hat{A}$).
$\hat{a}^{\mathrm{T}}\hat{b}$ denotes the $q$-dimensional inner product.
$\otimes$ denotes the Kronecker product.
\subsection{Starting Point}
\label{sec:startingPoint}
Consider a random variable $\mathbf{X}\in{\mathbb R}^n$ whose
probability density function (pdf) is a Gaussian mixture with equal weights.
The Gaussian pdfs all have covariance matrix $\sigma^2 {\bf 1}$, and they are centered on points
$\vecw_1,\ldots,\vecw_q \in {\mathbb R}^n$.
We write $\hat\vecw=(\vecw_1,\ldots,\vecw_q)$.
\begin{equation}
f_{\vecX| {\hat \vecW}}(\vecx|\hat\vecw)
= \frac1q\sum_{j=1}^q g_{\vecw_j,\sigma^2}(\vecx)
= (2\pi\sigma^2)^{-\frac n2} \frac1q \sum_{j=1}^q \exp \left\{- \frac{(\vecx-\vecw_j)^2}{2\sigma^2}\right\}.
\label{densityXgivenW}
\end{equation}
We note that a simplification is possible when $n>q$, i.e., the dimension of the space is higher than the number of components in the mixture. For all configurations $\hat\vecw$, it is possible to find a rotation such that after rotation $\vecw_j = \sum_{\alpha=1}^q w_{j\alpha} \vece_\alpha$, where the $\vece_\alpha$ are the basis vectors in ${\mathbb R}^n$. In other words, for all $\vecw_j$, the vector components in the dimensions beyond $q$ vanish. The density (\ref{densityXgivenW}) simplifies to
\begin{equation}
f_{\vecX| {\hat \vecW}}(\vecx|\hat\vecw) \to
\exp\left\{-\frac1{2\sigma^2} \sum_{\alpha=q+1}^n x_\alpha^2\right\}\;
(2\pi\sigma^2)^{-\frac n2} \frac1q \sum_{j=1}^q \exp \left\{- \frac{\sum_{\alpha=1}^q(x_\alpha-w_{j\alpha})^2}{2\sigma^2}\right\}.
\end{equation}
The first exponent is trivially dealt with when the $\ln f_{\vecX| {\hat \vecW}}(\vecx|\hat\vecw)$ is integrated, yielding a constant contribution $\frac{n-q}2$ to the entropy. Hence, we only focus on the case $n\leq q$ in our calculations. In general,
analytically approximating the differential entropy $h(\vecX|\hat\vecW=\hat\vecw)$ for a given $\hat\vecw$ is a difficult problem. 
In this paper, we study a problem that is slightly easier:
finding analytic approximations for $h(\vecX|\hat\vecW)={\mathbb E}_{\hat\vecw} h(\vecX|\hat\vecW=\hat\vecw)$
in the case where the offsets $\vecw_j$ are i.i.d. Gaussian-distributed with zero mean and covariance matrix $s^2{\bf 1}$.
\begin{equation}
f_{\hat\vecW}(\hat\vecw)
= \prod_{j=1}^q g_{0,s^2 {\bf 1}}(\vecw_j)
= (2\pi s^2)^{-\frac {nq}2} \exp\left\{ - \sum_{j=1}^q\frac{\vecw_j^2}{2s^2}\right\}.
\end{equation}
We note that the quantity $h(\vecX|\hat \vecW)$, under precisely these conditions, occurs in the field of spread-spectrum watermarking \cite{CKLS1997,WHGZL2022}.
In this setting, the $\hat\vecW$ is a table of $q$ random watermarking sequences.
A row index $K\in\{1,\ldots,q\}$ is chosen at random, and the $K$'th row of $\hat\vecW$
is additively inserted into some data
$\vecD$, which are modeled as Gaussian, and then attackers apply noise in order to wipe out the watermark.
The $\vecX$ plays the role of the watermarked data after this attack.
A watermark detector tries to determine the index $K$ based on $\vecX$ and $\hat\vecW$
(and optionally the unwatermarked original data~$\vecD$).
In the analysis of the detection efficiency, an important figure of merit is the mutual information $I(K;\vecX \hat\vecW)$
or $I(K;\vecD \vecX \hat\vecW)$,
both of which after some rewriting involve precisely the conditional entropy~$h(\vecX|\hat\vecW)$.
We take $s^2< \sigma^2$ and develop a power series in the parameter $\mu=s^2/\sigma^2$.
In the following sections, we obtain our results by following these steps:
\begin{itemize}
\item We first obtain a more compact form of $h(\vecX|\hat\vecW)$ in Corollary \ref{entropcoll} via a change of variables {in order to write $h(\vecX|\hat\vecW)$ in terms of the more familiar expectation with regard to a multivariate Gaussian density}. 

\item {We next perform another change or variables to obtain a diagonal form in the expectation with regard to the multivariate Gaussian. This is important since our new variables will be independent, and many expressions will be simplified if our variables are not correlated.} 
\item {We evaluate the leading terms to $h(\vecX|\hat\vecW)$, as shown in Theorem \ref{entrocaltheo}, where we obtain an expression that separates the leading contributions to $h(\vecX|\hat\vecW)$ and the quantity $S$ to be defined. The reason this is important is that the leading contributions contain terms that will make the Taylor series diverge, and therefore, we evaluate them analytically before performing the Taylor series expansion. We also include other terms that do not cause any problems in the limit of small $\mu$ for convenience. The remaining expression $S$ is then safe to expand for small $\mu$.}
\item We perform a third change of variables to simplify $S$, and we obtain a Taylor series for $h(\vecX|\hat\vecW)$ up to order $\mathcal{O}(\mu^2)$ in Theorem \ref{bruttheo} via a brute-force approach. {The change of variables is necessary for making the analytical expression tractable. }
\item Finally, we provide a determinant-based approach to evaluate the power series for $S$, and we obtain a result for $h(\vecX|\hat\vecW)$ up to order $\mathcal{O}(\mu)$.
\end{itemize}
\subsection{First Change of Variables}
We observe that inside the logarithm, the variables $\mathbf w_j$ occur only in the combination $\mathbf w_j -\mathbf x$. This allows us to introduce shifted variables $\mathbf z_j = \mathbf w_j -\mathbf x$ and then analytically carry out the integration over~$\mathbf x$.
\begin{lemma}
\label{borislemma}
The differential entropy $h(\vecX|\hat\vecW)$ is given by
\begin{equation}
\begin{aligned}
h(\vecX|\hat\vecW)=
\frac n2\ln 2\pi\sigma^2
-\left(\frac\mu q\right)^{\frac n2} (2\pi)^{-\frac{qn}2}  \frac1q\sum_{j=1}^q \int \rd\mathbf{z}_1 \cdots \rd\mathbf{z}_q\;
e^{-\frac12\sum_a \vecz_a^2  +\frac1{2q} (\sum_a \vecz_a)^2 -\frac\mu2\vecz_j^2}
\ln \frac1q \sum_{k=1}^q e^{-\frac\mu2 \vecz_k^2}.
\end{aligned}
\end{equation}
\end{lemma}
\begin{proof}
See Appendix \ref{borisapp}.
\end{proof}
Next, we can get rid of the $\sum_j$ summation, using permutation symmetry.
\begin{lemma}
\label{flemma}
Let $F(\mathbf{z}_1,\dots,\mathbf{z}_q)$ be any permutation invariant function; then, we have
\begin{equation}
\begin{aligned}
&\int \exp\left\{-\dfrac{1}{2}\sum_j \mathbf{z}_j^2+\dfrac{1}{2q}\sum_{ij}\mathbf{z}_i \cdot \mathbf{z}_j\right\}\dfrac{1}{q}\sum_k\exp\left\{-\dfrac{\mu}{2}\mathbf{z}_k^2 \right\} F(\mathbf{z}_1,\dots,\mathbf{z}_q)\mathrm{d}\mathbf{z}_1 \cdots \mathrm{d}\mathbf{z}_q\\
&=\int\exp\left\{-\dfrac{1}{2}\mathbf{\hat{z}}^{\mathrm{T}}\hat{M}^{(q)} \mathbf{\hat{z}}\right\} F(\mathbf{z}_1,\dots,\mathbf{z}_q)\mathrm{d}\mathbf{z}_1 \cdots \mathrm{d}\mathbf{z}_q,
\end{aligned}
\end{equation}
where the matrix elements $\{M^{(k)}_{ij}\}_{k=1}^{q}$ are given by\vspace{6pt}
\begin{equation}
M^{(k)}_{ij}=\delta_{ij}-\dfrac{1}{q}+\mu \delta_{ik}\delta_{jk}.
\end{equation}
\end{lemma}
\begin{proof}
See Appendix \ref{fexpecapp}.
\end{proof}
We now rewrite the differential entropy more compactly in the following corollary.
\begin{corollary}
\label{entropcoll}
Let the function $F$ be given by
\begin{equation}
\label{feq}
F(\mathbf{z}_1,\dots,\mathbf{z}_q)=\ln q^{-1}\sum_{l=1}^q \exp \left\{-\dfrac{\mu}{2}\mathbf{z}_l \cdot \mathbf{z}_l \right\},
\end{equation}
then, the differential entropy $h(\vecX|\hat\vecW)$ can be written as
\begin{equation}
\label{diffentcompact}
h(\vecX|\hat\vecW)=\frac n2\ln 2\pi\sigma^2
-\left(\frac\mu q\right)^{\frac n2} (2\pi)^{-\frac{qn}2}  \int \exp\left\{-\dfrac{1}{2}\mathbf{\hat{z}}^{\mathrm{T}}\hat{M}^{(q)} \mathbf{\hat{z}}\right\} F(\mathbf{z}_1,\dots,\mathbf{z}_q)\mathrm{d}\mathbf{z}_1 \cdots \mathrm{d}\mathbf{z}_q.
\end{equation}
\end{corollary}
\begin{proof}
It follows directly from Lemmas \ref{borislemma} and  \ref{flemma}, where we note that Equation \eqref{feq} is permutation invariant.
\end{proof}
\subsection{Diagonalization of $\hat{M}^{(q)}$ and Second Change of Variables}

We would like to obtain an expression for the expectation term in Equation \eqref{diffentcompact}; however, we first need to switch to a diagonal form that removes correlation between the integration variables, i.e., we would like to find the eigenvectors and eigenvalues of the matrix $\hat{M}^{(q)}$. We first notice that adding the identity matrix, while it shifts the eigenvalues by 1, does not change the eigenvectors of a matrix. Hence, we diagonalize the matrix $\hat{C}^{(q)}$, as shown in Theorem \ref{ctheo}, and our result follows directly in Corollary \ref{mcorr}.
\begin{theorem}
\label{ctheo}
Let $\hat{J}$ be the all-one matrix.
Let $\hat{e}^{(k)}$ denote the $k$-th standard basis vector.
Then, the matrix $\hat{C}^{(q)}$ given by
\begin{equation}
\hat{C}^{(q)} = -q^{-1}\hat{J}+\mu \hspace{0.5mm} \mathrm{diag}(\hat{e}^{(q)})
=-\dfrac{1}{q}\begin{pmatrix}
1 & 1 & \cdots & 1 \\
1 & 1 & \cdots & 1 \\
\vdots  & \vdots  & \ddots & \vdots  \\
1 & 1 & \cdots & 1-q\mu
\end{pmatrix}
,
\end{equation}
has an orthogonal matrix $\hat{\Lambda}^{(q)}$ of eigenvectors  given by
\begin{equation}
\hat{\Lambda}^{(q)}=(\hat{\alpha}_1, \hat{\alpha}_2, \dots, \hat{\alpha}_{q-2}, \hat{\beta}_1, \hat{\beta}_2),
\end{equation}
where $\{\hat{\alpha}_i\}^{q-2}_{i=1}$ are eigenvectors with eigenvalue 0, whereas $\hat{\beta}_1$ and $\hat{\beta}_2$ have eigenvalues $\lambda_1$ and $\lambda_2$, respectively, and read
\begin{equation}
\lambda_j=\dfrac{1}{2}\left((\mu-1)+(-1)^{j+1}\sqrt{\mu^2-(4 q^{-1}-2)\mu+1}\right), \quad j\in \{1,2\}.
\end{equation}
Furthermore, the elements of the eigenvectors are given by
\begin{align}
&\alpha_{i, k}=
\dfrac{1}{\sqrt{i(i+1)}}\begin{cases}0&, k>i+1\\1&, k<i+1\\-i&, k=i+1 \end{cases},\\
&\beta_{i, k}=
\dfrac{1}{\sqrt{q-1+[1-q(\lambda_{i}+1)]^2}}\begin{cases}1&, k\neq q\\1-q(\lambda_{i}+1)&, k=q \end{cases}.
\end{align}
\end{theorem}
\begin{proof}
See Appendix \ref{diagproof}.
\end{proof}
Note that in matrix notation, we can write
\begin{equation}
\label{mmatrix}
\hat{M}^{(q)}=\hat{I}-q^{-1}\hat{J}+\mu \hspace{0.5mm} \mathrm{diag}(\hat{e}^{(q)})=\hat{I}+\hat{C}^{(q)},
\end{equation}
where $\hat{I}$ denotes the identity matrix, and we have the following corollary:
\begin{corollary}
\label{mcorr}
The matrix $\hat{M}^{(q)}$ given by Equation \eqref{mmatrix} is diagonalizable by the matrix of eigenvectors $\hat{\Lambda}^{(q)}$ given in Theorem \ref{ctheo}, and we have
\begin{align}
\left(\hat{\Lambda}^{(q)}\right)^{\mathrm{T}}\hat{M}^{(q)}\hat{\Lambda}^{(q)}=\begin{pmatrix}
1 & 0 & \cdots & 0 &0\\
0 & 1 & \cdots & 0 &0\\
\vdots  & \vdots  & \ddots & \vdots&\vdots  \\
0 & 0 & \cdots & m_1&0\\
0&0&\cdots&0&m_2
\end{pmatrix}\equiv \hat{D}^{(q)},
\end{align}
where $m_1$ and $m_2$ are defined as
\begin{equation}
m_j \equiv \lambda_j+1,\quad j \in \{1,2\}.
\end{equation}
\end{corollary}
\begin{proof}
Follows from Theorem \ref{ctheo}.
\end{proof}
\begin{lemma}
\label{uvarlemm}
Let $\hat{\mathbf{z}}=\hat{\Lambda}^{(q)}\hat{\mathbf{u}}$, then the function $F$ given in Corollary \ref{entropcoll} can be written as
\begin{equation}
\begin{aligned}
F(\hat{\Lambda}^{(q)}\hat{\mathbf{u}})=&
-\dfrac{\mu}{2}\left(\dfrac{1}{||\hat{x}_1||^2} \mathbf{u}_{q-1} \cdot \mathbf{u}_{q-1}+\dfrac{2}{||\hat{x}_1||||\hat{x}_2||} \mathbf{u}_{q-1} \cdot \mathbf{u}_{q}+\dfrac{1}{||\hat{x}_2||^2} \mathbf{u}_{q} \cdot \mathbf{u}_{q}\right)\\
&+\ln q^{-1} \sum_l \exp \left\{-\dfrac{\mu}{2}\left(G_l+2\delta_{q, l}T\right)\right\},
\end{aligned}
\end{equation}
where we have
\footnotesize
\begin{align}
\label{Geq}
&\hat{x}_j=(1,\dots,1,x_j)^{\mathrm{T}},\quad j \in\{1,2\},\\
&x_j=1-qm_j,\quad j \in\{1,2\},\\
&G_l= \sum_{m=1}^{q-2}\sum_{m'=1}^{q-2} {\alpha}_{m, l}{\alpha}_{m', l} \mathbf{u}_m \cdot \mathbf{u}_{m'}+\dfrac{2}{\sqrt{q-1+x_1^2}}\sum_{m=1}^{q-2} {\alpha}_{m, l} \mathbf{u}_m \cdot \mathbf{u}_{q-1}+\dfrac{2}{\sqrt{q-1+x_2^2}}\sum_{m=1}^{q-2}  \alpha_{m, l}\mathbf{u}_m \cdot \mathbf{u}_{q},\\
\label{Teq}
&T= \dfrac{x_2^2-1}{2(q-1+x_2^2)}\mathbf{u}_{q} \cdot \mathbf{u}_{q}+\dfrac{x_1^2-1}{2(q-1+x_1^2)}\mathbf{u}_{q-1} \cdot \mathbf{u}_{q-1}+\dfrac{x_1x_2-1}{\sqrt{q-1+x_1^2}\sqrt{q-1+x_2^2}}\mathbf{u}_{q} \cdot \mathbf{u}_{q-1}.
\end{align}
\normalsize
\end{lemma}
\begin{proof}
See Appendix \ref{uvarproof}.
\end{proof}
\subsection{Pulling Leading-Order Terms Out of the Integral}

We now use Corollary \ref{entropcoll} and Lemma \ref{uvarlemm} to further obtain an expression for the differential entropy that separates the leading contributions and the quantity $S$, which we define as~follows:
\begin{equation}
\label{Sdef}
S \equiv -\mathbb{E}_{u}\left[\ln  q^{-1}\sum_l \exp \left\{-\dfrac{\mu}{2}\left(G_l+2\delta_{q, l}T\right)\right\}\right],
\end{equation}
where $\mathbb{E}_u$ is the expectation taken with respect to the density $\prod_{i=1}^{n}\prod_{k=1}^{q}g_{0, d^{-1}_{k}}(u_{k, i})$, and $\{d_k\}_{k=1}^{q}$ are given by:
\begin{equation}
d_{k}= \begin{cases}1& k \leq q-2\\
m_1& k=q-1\\
m_2& k=q\end{cases}.
\end{equation}
\begin{theorem}
\label{entrocaltheo}
Let $h_\sigma=\ln \sigma\sqrt{2\pi e}$ be the differential entropy of a Gaussian distributed random variable with variance $\sigma$.
The differential entropy $h(\vecX|\hat\vecW)$ is given by
\begin{equation}
\label{hextractS}
\begin{aligned}
h(\vecX|\hat\vecW)=n h_\sigma+\dfrac{n}{2}\dfrac{q}{q-1}\mu+S,
\end{aligned}
\end{equation}
where $S$ is given by Equation \eqref{Sdef}.
\end{theorem}
\begin{proof}
See Appendix \ref{entapp}.
\end{proof}
\subsection{Third Change of Variables: Simplification} 
\label{tildevar}
We would like to approximate $S$ given in Equation \eqref{Sdef}. However, our current integration variables $\{\mathbf{u}_k\}_{k=1}^q$ do not give simple formulas for $G_l$, $T$, or $S$. By linearly mixing $\mathbf{u}_{q-1}$ and $\mathbf{u}_q$, we obtain a set of $q$-independent standard normal-distributed variables with the benefit that $G_l$ contains only one of the two newly introduced variables. We define the transformation matrix $W$ as
\begin{equation}
\label{wmat}
W \equiv \dfrac{1}{\sqrt{\dfrac{1}{\mu}+\dfrac{q}{q-1}}}\begin{pmatrix}
\dfrac{1/\sqrt{m_2}}{\sqrt{q-1+x_2^2}} & \dfrac{1/\sqrt{m_1}}{\sqrt{q-1+x_1^2}} \\
\dfrac{1/\sqrt{m_1}}{\sqrt{q-1+x_1^2}} & -\dfrac{1/\sqrt{m_2}}{\sqrt{q-1+x_2^2}} \\\end{pmatrix},
\end{equation}
where note that we can use Equation \eqref{simpli2} to show that $\mathrm{det}(W)=-1$.
\begin{lemma}
\label{wlemm}
Let $\mathbf{r}=\sqrt{m_2}\mathbf{u}_{q}$, $\mathbf{s}=\sqrt{m_1}\mathbf{u}_{q-1}$, and $W$ be as given in Equation \eqref{wmat}, then the variables $\tilde{\mathbf{r}}$ and $\tilde{\mathbf{s}}$ given by
\begin{equation}
\begin{pmatrix}\mathbf{r}\\ \mathbf{s}\end{pmatrix}=
W
\begin{pmatrix}\tilde{\mathbf{r}}\\\tilde{\mathbf{s}}\end{pmatrix}
\end{equation}
have independent standard normal distributions, i.e., we can write
\begin{equation}
\begin{aligned}
&\mathbb{E}_u[\cdot]=\mathbb{E}_{u,\tilde{s}, \tilde{r}}[\tilde{\cdot}],
\end{aligned}
\end{equation}
where $[\tilde{\cdot}]$ denotes the transformed expression, and $\mathbb{E}_{u,\tilde{s}, \tilde{r}}[\tilde{\cdot}]$ is taken with regard to the density\\
$\prod_{a=1}^{n}g_{0, 1}(\tilde{r}_a)g_{0, 1}(\tilde{s}_a)\prod_{b=1}^{q-2}g_{0,1}(u_{b, a})$.
\end{lemma}
\begin{proof}
See Appendix \ref{tildproof}.
\end{proof}
As a corollary, $T$ and $G_l$ assume a simpler form in terms of our new variables.
\begin{corollary}
\label{TGintildcorr}
Let $\mathbf{w}_l$ be given by
\begin{equation}
\mathbf{w}_l=\sum_{m=1}^{q-2}\alpha_{m, l} \mathbf{u}_{m},
\end{equation}
then after applying the change of variables in Lemma \ref{wlemm}, $T$ and $G_l$ can be written as
\begin{align}
\label{Gnewv}
&G_l=\mathbf{w}_l \cdot \mathbf{w}_l+2\sqrt{\dfrac{1}{\mu}+\dfrac{q}{q-1}} \mathbf{w}_l \cdot \tilde{\mathbf{r}},\\
\label{Tnewv}
&T=\dfrac{1}{2}\left(1+\mu\dfrac{q}{q-1}\right)^{-1}\left[-\dfrac{q}{q-1}\left(2+\mu\dfrac{q}{q-1}\right)\tilde{\mathbf{r}}\cdot \tilde{\mathbf{r}}+\dfrac{q}{q-1}\tilde{\mathbf{s}}\cdot \tilde{\mathbf{s}}+\dfrac{2}{\sqrt{\mu}}\sqrt{\dfrac{q}{q-1}}\tilde{\mathbf{r}}\cdot \tilde{\mathbf{s}}\right].
\end{align}
\end{corollary}
\begin{proof}
It follows by direct substitution.
\end{proof}
In this new {form} 
(\ref{Gnewv}) and (\ref{Tnewv}), the powers of $\mu$ are explicitly visible.
The variables $\mathbf{w}_l,\tilde{\mathbf r},\tilde{\mathbf s}$
do not generate any $\mu$-dependence when integrated.
It is clear that $\mu G_l$ and $\mu T$ are of order $\mathcal{O}(\sqrt\mu)$;
hence, we can formally proceed with a Taylor expansion of the $\mathrm{exp}$ function in (\ref{Sdef}).

\section{Brute Force Expansion}

We develop a series expansion in $\mu$ as follows. First, we expand the exp function in~(\ref{Sdef}), with $G_l$ and $T$ as {defined in} (\ref{Gnewv}) and (\ref{Tnewv}). This yields a complicated series, whose leading order behaviour is $1+{\cal O}(\sqrt\mu)$. Then, we substitute that series into the Taylor expansion of $\ln(1+\varepsilon)$ in order to evaluate~$S$ as a series expansion. Finally, we apply the expectation $\mathbb{E}_{u, \tilde{s}, \tilde{r}}$, which is possible because we only obtain expectations of powers of $\mathbf u,\tilde{\mathbf r},\tilde{\mathbf s}$, which are independent normal-distributed variables. 
Although all steps by themselves are straightforward enough, the result is a {\em huge} amount of bookkeeping. A number of things are worth noting:
\begin{itemize}
\item
What helps in this exercise is that odd powers of $\mathbf w_l$,  $\tilde{\mathbf r}$, and $\tilde{\mathbf s}$
lead to vanishing integrals.
\item
Furthermore, it is clear from the start that the odd powers of $\sqrt\mu$ will disappear in the end.
This can be seen from the fact that a factor $\frac1{\sqrt\mu}$ in $T$ and $G_l$ always
occurs with an isolated $\tilde{\mathbf r}$;
hence, any occurrence of an odd power of $\sqrt\mu$ comes with an odd power of~$\tilde{\mathbf r}$.
\item
Due to the $n$-dimensional inner products that occur in $G_l$ and $T$,
which consist of $n$ independent terms, each power of $\mu$ is associated with a factor~$n$.
Consequently, the power series becomes a series not in $\mu$ but actually in $n\mu$.
For convergence, the product $n\mu$ needs to be sufficiently small.
Fortunately, we are allowed to work under the condition $n\leq q$, as explained in Section \ref{sec:startingPoint}.
\end{itemize}

We performed this brute-force exercise up to order $\mathcal{O}(\mu^2)$.
\begin{theorem}
\label{bruttheo}
The differential entropy $h(\vecX|\hat\vecW)$ up to order $\mathcal{O}(\mu^2)$ is given by
\begin{equation}
\begin{aligned}
\label{bruteq}
h(\vecX|\hat\vecW)=nh_{\sigma}+\dfrac{n}{2}\left(1-\dfrac{1}{q}\right)\mu-\dfrac{n}{2}\left(1-\dfrac{1}{q}\right)\dfrac{nq^{-1}+1}{2q}\mu^2+\mathcal{O}\left(\mu^{3}\right).
\end{aligned}
\end{equation}
\end{theorem}
\begin{proof}
See Appendix \ref{bruteforceproof}.
\end{proof}
Theorem \ref{bruttheo} is not directly comparable with results in the literature. However, we can make use of existing upper bounds in order to assess our expression. The Gaussian upper bound in Equation \eqref{boundstup} is known to be loose, while the bound in \eqref{betterbound} is closer to the value of the differential entropy. For our case, the bound in Equation \eqref{betterbound} reads
\begin{equation}
\label{betterbound2}
h(\vecX|\hat\vecW) \leq nh_{\sigma}+\ln q.
\end{equation}
It is easy to see that for $\mu \leq q^{-1}$ and $n=q$, our second-order Taylor expansion in Equation~\eqref{bruteq} satisfies this bound up to the error produced by the truncation of the series, which is of order~$o(\mu^2)$.
\section{Determinant-Based Approach}
We introduce a less brute-force series expansion that does not expand the exp function but focuses on expanding the logarithm.
This gives a power series in exponential expressions; each of these expressions can be integrated analytically, yielding a matrix determinant. For simplicity of notation we write the variables $\mathbf u_k,\tilde{\mathbf r},\tilde{\mathbf s}$ together as a $q$-component vector $\hat{\mathbf a}$,
\begin{equation}
\label{renamevar}
\hat\veca= (\vecu_1,\ldots,\vecu_{q-2},\tilde\vecr,\tilde\vecs).
\end{equation}
The $S$ is now written as follows:
\begin{equation}
\label{Seq}
S=-{\mathbb E}_{a} \left[\ln q^{-1}\sum_{l=1}^q \exp \left\{-\dfrac{\mu}{2}\left(G_l+2\delta_{q, l}T\right)\right\}\right],
\end{equation}
where we have
\begin{equation}
\label{expectointtild}
\mathbb{E}_{a}\left[\cdots\right]=\left(\sqrt{2\pi}\right)^{-nq}\int \exp\left\{-\dfrac{1}{2}\mathbf{\hat{a}}^{\mathrm{T}}\hat{I} \mathbf{\hat{a}}\right\}  \left[\cdots\right]\mathrm{d}\mathbf{a}_1 \cdots \mathrm{d}\mathbf{a}_q.
\end{equation}
Next, we can write $G_l$ and $T$ in vector-matrix-vector product form:
\begin{equation}
\label{GmatTmat}
\mu G_{l}=\hat{\mathbf{a}}^\mathrm{T} \hat{Q}_l \hat{\mathbf{a}},
\quad\quad\quad
2\mu T=\hat{\mathbf{a}}^\mathrm{T}\hat{P} \hat{\mathbf{a}},
\end{equation}
where the matrix elements of $\hat{Q}_l$ and $\hat{P}$ are given by
\begin{equation}
\label{qmatele}
\begin{aligned}
Q_{l,mm'} = \mu \omega_{m, l}\omega_{m', l}+ \sqrt\mu \sqrt{1
+\mu\frac{q}{q-1}} [ \omega_{m, l} \delta_{m', q-1}  +  \omega_{m', l} \delta_{m, q-1}] ,
\end{aligned}
\end{equation}
where we define
\begin{equation}
\label{omele}
\omega_{m, l} \equiv \begin{cases}\alpha_{m, l}&, m \leq q-2\\
0&, \text{otherwise}\\
\end{cases}
.
\end{equation}
\begin{equation}
\label{pmatele}
\begin{aligned}
P_{m,m'} = \left(1+\mu\frac{q}{q-1}\right)^{-1}
\begin{cases}
\mu\frac{q}{q-1}&, m=m'=q\\
-\mu\frac{q}{q-1}\left(2+\mu\frac{q}{q-1}\right)&, m=m'=q-1\\
\sqrt{\mu}\sqrt{\frac{q}{q-1}}&, m=q, m'=q-1\\
\sqrt{\mu}\sqrt{\frac{q}{q-1}}&, m=q-1, m'=q\\
0&, \text{otherwise}
\end{cases}
.
\end{aligned}
\end{equation}
The $\hat P$ is essentially a $2\times 2$ matrix.
The $\hat Q_\ell$ consists of the projector $\hat \alpha_\ell \hat \alpha_\ell^T$
plus extra entries in row and column~$q-1$.
Let $Z$ be defined as
\begin{equation}
\begin{aligned}
Z &= 1-q^{-1}\sum_{l=1}^q \exp \left\{-\frac{\mu}{2}\left(G_l+2\delta_{q, l}T\right)\right\}\\
&=1-q^{-1}\sum_{l=1}^{q-1} \exp \left\{-\frac{1}{2}\hat{\mathbf{a}}^\mathrm{T} \hat{Q}_l \hat{\mathbf{a}}\right\}-q^{-1}\exp \left\{-\frac{1}{2}\hat{\mathbf{a}}^\mathrm{T} \hat{P} \hat{\mathbf{a}}\right\},
\end{aligned}
\end{equation}
where equality follows from the definitions \eqref{GmatTmat}.
We note that $Z={\cal O}(\sqrt\mu)$.
This allows us to perform an expansion in $Z$, knowing that $Z^k$ does not produce powers of $\mu$
lower than $\mu^{k/2}$. 
Furthermore, the power $\sqrt\mu$ occurs only in the off-diagonal components of $\hat P$
and $\hat Q_\ell$, whereas the diagonal components are~${\cal O}(\mu)$.
As in the brute-force approach, the integration over $\hat\veca$ eliminates half-integer powers of~$\mu$.
Hence, $\EE_a Z^k = {\cal O}(\mu^{\lceil k/2\rceil})$.
In order to obtain the entropy estimation up to and including power
$\mu^t$, we need all contributions up to and including $\EE_a Z^{2t}$. Using $\ln (1-Z)=-\sum_{k=1}^{\infty}k^{-1}Z^k$, one can easily show that $S$~(\ref{Seq}) can be written as
\begin{equation}
S=\sum_{k=1}^{\infty}\mathcal{Z}_k, \quad\quad
{\cal Z}_k = \frac1k \EE_a Z^k.
\end{equation}
This provides a recipe for going to arbitrary order in~$\mu$.
We show the calculation of the $\mu^1$ contribution in $S$.
For this, we need only ${\cal Z}_1$ and ${\cal Z}_2$.
\begin{lemma}
\label{gausdetlemm}
For a positive definite matrix $\hat{V}$, we have

$
\int \exp\left\{-\dfrac{1}{2}\mathbf{\hat{u}}^{\mathrm{T}}\hat{V} \mathbf{\hat{u}}\right\}  \mathrm{d}\mathbf{u}_1 \cdots \mathrm{d}\mathbf{u}_q=\dfrac{(2\pi)^{nq/2}}{\det(\hat{V})^{n/2}}
$.
\normalsize
\end{lemma}
\begin{proof}
See Appendix \ref{gausdetproof}.
\end{proof}
\begin{lemma}
\label{ssereieslemm}
It holds that
\begin{align}
\label{z1cont}
&\mathcal{Z}_1=1-\sum_{\ell=1}^{q-1}\dfrac{q^{-1}}{\det(\hat{I}+\hat{Q}_\ell)^{n/2}}-\dfrac{q^{-1}}{\det(\hat{I}+\hat{P})^{n/2}},
\\
\nonumber
&\mathcal{Z}_2= \frac12 -\sum_{\ell=1}^{q-1}\frac{q^{-1}}{\det(\hat{I}+\hat{Q}_\ell)^{n/2}}-\frac{q^{-1}}{\det(\hat{I}+\hat{P})^{n/2}}+\sum_{\ell=1}^{q-1}\frac{q^{-2}}{\det(\hat{I}+\hat{Q}_\ell+\hat{P})^{n/2}}\\
\label{z2cont}
&\quad\quad
+\sum_{\ell=1}^{q-1}\sum_{\ell'=1}^{q-1} \frac{(1/2)q^{-2}}{\det(\hat{I}+\hat{Q}_\ell+\hat{Q}_{\ell'})^{n/2}}+\frac{(1/2)q^{-2}}{\det(\hat{I}+2\hat{P})^{n/2}}.
\end{align}
\end{lemma}
\begin{proof}
See Appendix \ref{Zcalcapp}.
\end{proof}
\begin{proposition}
For $t\in{\mathbb N}$, $\ell'\neq\ell$
\begin{eqnarray}
\det(\hat I+t\hat P) &=& 1-\mu t(t+1) \frac q{q-1},
\\
\det(\hat I+t\hat Q_\ell) &=& 1 - \mu t(t-1)\frac{q-2}{q-1}
-t^2 \mu^2 \frac{q(q-2)}{(q-1)^2},
\\
\det(\hat I+\hat P+\hat Q_\ell) &=& 1-\mu\frac{2q}{q-1}-\mu^2 \frac{4q(q-2)}{(q-1)^2},
\\
\det(\hat I+\hat Q_\ell+\hat Q_{\ell'}) &=&
1+\mu\frac{2}{q-1} -\mu^2 \frac{(3q-1)(q-3)}{(q-1)^2}
-\mu^3    \frac {2q(q-3)}{(q-1)^2}.
\end{eqnarray}
\end{proposition}
\begin{corollary}
It holds that
\begin{equation}
{\cal Z}_1 = -\frac n{q-1}\mu +{\cal O}(\mu^2),
\quad\quad\quad
{\cal Z}_2 = \frac {n/2}{q(q-1)}\mu +{\cal O}(\mu^2).
\end{equation}
\end{corollary}
Using (\ref{hextractS}) it follows that $h(\vecX|\hat\vecW)=nh_\sigma+\frac n2\frac q{q-1}\mu+S$,
with $S\approx {\cal Z}_1+{\cal Z}_2$, yielding
$h(\vecX|\hat\vecW)=nh_\sigma+\frac n2(1-\frac1q)\mu +{\cal O}(\mu^2)$.
This is consistent with the result of the brute-force approach, as in
Theorem~\ref{bruttheo}. In higher orders of $Z$, determinants occur of the form
$\det(\hat I + r \hat P + t \hat Q_\ell + t' \hat Q_{\ell'} + t'' \hat Q_{\ell''}+\cdots )$
for integer $r,t,t',t'',\ldots$.
These are computable but involve a lot of bookkeeping.
Because of the form $[\det(\cdots)]^{-n/2}$ and the fact that each determinant starts with 1,
followed by integer powers of $\mu$,
the leading order coefficient of $\mu^t$ is proportional to $\binom{n/2}{t}\propto n^t$.
Hence, as in the brute force approach,
the expansion in powers of $\mu$ effectively becomes an expansion in powers of~$n\mu$.
\section{Discussion}
Gaussian mixtures occur in the literature either as an approximation to a probability density function or as the actual density of a two-stage stochastic process. Shannon entropy plays an important role in the signal processing, optimization, and analysis of systems. It is important to have (semi-)analytical methods for evaluating the Shannon entropy (differential entropy) of Gaussian mixtures, without lengthy numerical computations in high-dimensional spaces. We have developed an expansion method that has a number of advantages: (i) it yields an entropy estimate rather than merely bounds; (ii) the order of magnitude of the error is known; (iii) it avoids the process of \cite{approx} that splits wide Gaussians into narrow pieces, which leads to a large number of expansion terms and introduces inaccuracies. The inaccuracies in \cite{approx} grow with the number of wide Gaussians in the mixture. \textcolor[rgb]{0.00,0.00,0.00}{The approach in \cite{approx} is unable to exploit the existence of a small parameter $\mu$ when the mixture structure allows it.} \textcolor[rgb]{0.00,0.00,0.00}{In our case, we make sure to leverage the trade-off between the support shared between the components and the variance of the components, which makes our method more efficient when the mixture is not composed of separated clusters. In the case of negligible shared support between the components, an approximation is not needed as the bound in Equation \eqref{betterbound2} becomes exact.} From the second-order result (\ref{bruteq}), it seems that convergence is faster than what we expected. In leading order terms, every power of $\mu$ is accompanied by a power of $n$, which may endanger the radius of convergence. However, in the $\mu^2$ part in (\ref{bruteq}), we see that the additional power of $n$ comes as $n/q$, which is not dangerous since we are allowed to work with $n\leq q$. Setting $n=q$ and $q\gg 1$,~(\ref{bruteq}) reduces to $h(\vecX|\hat\vecW)\approx \frac n2[\ln 2\pi e \sigma^2+\mu-\frac1q \mu^2+{\cal O}(\mu^3)]$, which has fast convergence. In future work, we will try to apply our method for
fixed displacements $\hat\vecw$ and study the convergence at higher orders of~$\mu$.
\subsection*{Acknowledgements}
This work was supported by NWO grant CS.001 (Forwardt) and by the Dutch Groeifonds project Quantum Delta NL KAT-2.
\bibliographystyle{plain}
\bibliography{references}
\appendix
\section{Proof of Lemma \ref{borislemma}}
\label{borisapp}
\begin{equation}
\begin{aligned}
&h(\vecX|\hat\vecW) = - \int \rd\mathbf{w}_1 \cdots \rd\mathbf{w}_q f_{\hat\vecW}(\hat\vecw)
\int \rd\mathbf{x}\; f_{\vecX|\hat\vecW}(\vecx|\hat\vecw) \ln f_{\vecX|\hat\vecW}(\vecx|\hat\vecw)
\\
& =- \frac1q\sum_{j=1}^q \int \rd\mathbf{x}\; \int \rd\mathbf{w}_1 \cdots \rd\mathbf{w}_q \left[\prod_{a=1}^q {g}_{0,s^2}(\vecw_a)\right]
{g}_{0,\sigma^2}(\vecw_j-\vecx) \ln \frac1q \sum_{k=1}^q {g}_{0,\sigma^2}(\vecw_k-\vecx).
\end{aligned}
\end{equation}
\noindent All the integrals are from $-\infty$ to $\infty$. Inside the $x$-integration, we define new integration variables $\vecz_j$ by writing $ \vecw_j =\vecx - s \vecz_j$. Furthermore, we write $\vecx=s \vecc$. Now, we have
\begin{equation}
\nonumber
\begin{aligned}
& h(\vecX|\hat\vecW) =
- s^{qn+n}\frac1q\sum_{j=1}^q \int \rd\mathbf{c}\; \int \rd\mathbf{z}_1 \cdots \rd\mathbf{z}_q
\left[\prod_{a=1}^q {g}_{0,s^2}(s\vecc-s\vecz_a)\right]
{g}_{0,\sigma^2}(s\vecz_j) \ln \frac1q \sum_{k=1}^q {g}_{0,\sigma^2}(s\vecz_k)
\\ &=
- s^{qn+n}\frac1q\sum_{j=1}^q \int \rd\mathbf{z}_1 \cdots \rd\mathbf{z}_q\left\{\int \rd\mathbf{c}\;(2\pi s^2)^{-\frac{qn}2}
e^{-\frac12 \sum_a (\vecc-\vecz_a)^2} \right\}
{g}_{0,\sigma^2}(s\vecz_j) \ln \frac1q \sum_{k=1}^q {g}_{0,\sigma^2}(s\vecz_k)
\\ &=
-s^{n} \frac1q\sum_{j=1}^q \int \rd\mathbf{z}_1 \cdots \rd\mathbf{z}_q \left\{
(2\pi)^{-\frac{qn}2}e^{-\frac12\sum_a \vecz_a^2} \int\rd\mathbf{c}\; e^{-\frac q2\vecc^2 +\vecc\cdot \sum_a\vecz_a}
\right\}
{g}_{0,\sigma^2}(s\vecz_j) \ln \frac1q \sum_{k=1}^q {g}_{0,\sigma^2}(s\vecz_k)
\\ &=
-s^n   \left(\frac{2\pi}{q}\right)^{\frac n2}  (2\pi)^{-\frac{qn}2} \frac1q\sum_{j=1}^q \int \rd\mathbf{z}_1 \cdots \rd\mathbf{z}_q\;
e^{-\frac12\sum_a \vecz_a^2  +\frac1{2q} (\sum_a \vecz_a)^2 }
{g}_{0,\sigma^2}(s\vecz_j) \ln \frac1q \sum_{k=1}^q {g}_{0,\sigma^2}(s\vecz_k)
\end{aligned}
\end{equation}
\begin{equation}
\begin{aligned}
&=
-\left(\frac\mu q\right)^{\frac n2} (2\pi)^{-\frac{qn}2}  \frac1q\sum_{j=1}^q \int \rd\mathbf{z}_1 \cdots \rd\mathbf{z}_q\;
e^{-\frac12\sum_a \vecz_a^2  +\frac1{2q} (\sum_a \vecz_a)^2 -\frac\mu2\vecz_j^2}
\ln \frac1q \sum_{k=1}^q \frac{\exp \left\{-\frac\mu2 \vecz_k^2\right\}}{(2\pi \sigma^2)^{n/2}}
\\ &=
\frac n2\ln 2\pi\sigma^2
-\left(\frac\mu q\right)^{\frac n2} (2\pi)^{-\frac{qn}2}  \frac1q\sum_{j=1}^q \int \rd\mathbf{z}_1 \cdots \rd\mathbf{z}_q\;
e^{-\frac12\sum_a \vecz_a^2  +\frac1{2q} (\sum_a \vecz_a)^2 -\frac\mu2\vecz_j^2}
\ln \frac1q \sum_{k=1}^q e^{-\frac\mu2 \vecz_k^2}.
\end{aligned}
\end{equation}
\normalsize
\section{Proof of Lemma \ref{flemma}}
\label{fexpecapp}
\begin{equation}
\label{expec}
\begin{aligned}
&\int \exp\left\{-\dfrac{1}{2}\sum_j \mathbf{z}_j^2+\dfrac{1}{2q}\sum_{ij}\mathbf{z}_i \cdot \mathbf{z}_j\right\}\dfrac{1}{q}\sum_k\exp\left\{-\dfrac{\mu}{2}\mathbf{z}_k^2 \right\} F(\mathbf{z}_1,\dots,\mathbf{z}_q)\mathrm{d}\mathbf{z}_1 \cdots \mathrm{d}\mathbf{z}_q\\
&=\dfrac{1}{q}\sum_k\int \exp\left\{-\dfrac{1}{2}\sum_{ij} \delta_{ij}\mathbf{z}_i \cdot \mathbf{z}_j+\dfrac{1}{2q}\sum_{ij}\mathbf{z}_i \cdot \mathbf{z}_j-\sum_{ij}\dfrac{\mu}{2}\delta_{ik}\delta_{jk}\mathbf{z}_i \cdot \mathbf{z}_j\right\}F(\mathbf{z}_1,\dots,\mathbf{z}_q) \mathrm{d}\mathbf{z}_1 \cdots \mathrm{d}\mathbf{z}_q\\
&=\dfrac{1}{q}\sum_k\int \exp\left\{-\dfrac{1}{2}\sum_{ij}M^{(k)}_{ij} \mathbf{z}_i \cdot \mathbf{z}_j\right\} F(\mathbf{z}_1,\dots,\mathbf{z}_q)\mathrm{d}\mathbf{z}_1 \cdots \mathrm{d}\mathbf{z}_q\\
&=\dfrac{1}{q}\sum_k\int \exp\left\{-\dfrac{1}{2}\mathbf{\hat{z}}^{\mathrm{T}}\hat{M}^{(k)} \mathbf{\hat{z}}\right\} F(\mathbf{z}_1,\dots,\mathbf{z}_q)\mathrm{d}\mathbf{z}_1 \cdots \mathrm{d}\mathbf{z}_q.
\end{aligned}
\end{equation}
\noindent Note that we have the following:
\begin{equation}
\begin{aligned}
\mathbf{\hat{z}}^{\mathrm{T}}\hat{M}^{(k)} \mathbf{\hat{z}}=\sum_j \mathbf{z}_j^2-\dfrac{1}{q}\sum_{ij}\mathbf{z}_i \cdot \mathbf{z}_j+\mu\mathbf{z}_k^2=\mathbf{\hat{z}}'^{\mathrm{T}}\hat{M}^{(k')} \mathbf{\hat{z}}',
\end{aligned}
\end{equation}
where $\mathbf{\hat{z}}'$ are the following permutation of $\mathbf{\hat{z}}$:
\begin{equation}
\mathbf{{z}}'_i=\begin{cases}\mathbf{{z}}_i&, i\neq k, k'\\
\mathbf{{z}}_k&, i=k'\\
\mathbf{{z}}_{k'}&, i=k
\end{cases}
,
\end{equation}
therefore, for any $F$ that is permutation invariant, we can write
\small
\begin{equation}
\label{symm}
\dfrac{1}{q}\sum_k\int \exp\left\{-\dfrac{1}{2}\mathbf{\hat{z}}^{\mathrm{T}}\hat{M}^{(k)} \mathbf{\hat{z}}\right\} F(\mathbf{z}_1,\dots,\mathbf{z}_q)\mathrm{d}\mathbf{z}_1 \cdots \mathrm{d}\mathbf{z}_q=\int \exp\left\{-\dfrac{1}{2}\mathbf{\hat{z}}^{\mathrm{T}}\hat{M}^{(q)} \mathbf{\hat{z}}\right\} F(\mathbf{z}_1,\dots,\mathbf{z}_q)\mathrm{d}\mathbf{z}_1 \cdots \mathrm{d}\mathbf{z}_q.
\end{equation}
\normalsize
\section{Proof of Theorem \ref{ctheo}}
\label{diagproof}
Suppose some eigenvectors of $\hat{C}^{(q)}$ are of the form $(1,\dots,1,x)^{\mathrm{T}}$, where $x$ is unknown. We now wish to solve the eigenvalue equation:
\begin{equation}
-\dfrac{1}{q}\begin{pmatrix}
1 & 1 & \cdots & 1 \\
1 & 1 & \cdots & 1 \\
\vdots  & \vdots  & \ddots & \vdots  \\
1 & 1 & \cdots & 1-q\mu
\end{pmatrix}\begin{pmatrix}1\\ \vdots\\1\\x\end{pmatrix}=\lambda \begin{pmatrix}1\\ \vdots\\1\\x\end{pmatrix},
\end{equation}
for some eigenvalue $\lambda$. Carrying out the matrix multiplication, we obtain
\begin{equation}
\begin{pmatrix}q-1+x\\ \vdots\\q-1+x\\q-1+(1-q\mu)x\end{pmatrix}=\begin{pmatrix}-\lambda q\\ \vdots\\-\lambda q\\-\lambda q x\end{pmatrix},
\end{equation}
which give the following two equations:
\begin{align}
\label{xequation}
&q-1+x=-\lambda q \Leftrightarrow x=-q(\lambda+1)+1,   \\
&q-1+(1-q\mu)x=-\lambda q x \Leftrightarrow x=-\dfrac{q-1}{\lambda q-q\mu+1},
\end{align}
and by equating, we obtain
\begin{equation}
\begin{aligned}
&q(\lambda+1)-1=\dfrac{q-1}{\lambda q-q\mu+1} \Leftrightarrow q-1=\left(\lambda q-q\mu+1\right)\left(q(\lambda+1)-1\right)\\
&=q^2 \lambda(\lambda+1)-q\lambda-q^2\mu(\lambda+1)+q\mu+q(\lambda+1)-1\\
&=q^2\lambda^2+q^2\lambda-q\lambda-q^2\mu \lambda-q^2\mu+q\mu+q\lambda+q-1.
\end{aligned}
\end{equation}
We have the quadratic equation:
\begin{equation}
q^2\left(\lambda^2+(1-\mu)\lambda+\mu(q^{-1}-1)\right)=0.
\end{equation}
Since $q>0$, we have the solutions:
\begin{align}
\lambda_{1,2}=\dfrac{1}{2}\left((\mu-1)\pm\sqrt{D}\right),
\end{align}
where $D$ is given by
\begin{equation}
D\equiv (\mu -1)^2-4\mu(q^{-1}-1)=\mu^2-2\mu+1-4\mu q^{-1}+4\mu=\mu^2-(4 q^{-1}-2)\mu+1.
\end{equation}
From Equation \eqref{xequation}, obtain the two following values for $x$:
\begin{align}
&x_{1,2}=-qm_{1,2}+1,
\end{align}
where we write $m_1=\lambda_1+1$ and $m_2=\lambda_2+1$, and we define $\hat{x}_1$ and $\hat{x}_2$ by
\begin{align}
\hat{x}_1\equiv(1,\dots,x_1)^{\mathrm{T}},\quad \hat{x}_2\equiv(1,\dots,x_2)^{\mathrm{T}}.
\end{align}
Since $x_1$ and $x_2$ are solutions to a quadratic equation, we have the following simplifications:
\begin{align}\label{simpli}
&m_1 m_2 =\dfrac{\mu}{q},\\
&m_1 + m_2 =\mu +1, \\
&x_1 + x_2 = -q(\mu+1)+2, \\
\label{xx}
&x_1 x_2=1-q,\\
\label{simpli2}
&\dfrac{\mu/m_1}{q-1+x_1^2}+\dfrac{\mu/m_2}{q-1+x_2^2}=1+\mu\dfrac{q}{q-1},   \\
&\dfrac{(x_1^2-1)/m_1}{q-1+x_1^2}+\dfrac{(x_2^2-1)/m_2}{q-1+x_2^2}=\dfrac{q}{1-q},   \\
&\dfrac{(x_1^2-1)(x_2^2-1)-q^2}{(q-1+x_1^2)(q-1+x_2^2)}=\dfrac{1/m_2}{q-1+x_2^2}-\dfrac{1/m_1}{q-1+x_1^2}=\dfrac{1}{1-q},  \\
&\dfrac{x_1/m_1}{q-1+x_1^2}+\dfrac{x_2/m_2}{q-1+x_2^2}=\dfrac{1}{\mu}.
\label{prodsimp}
\end{align}
Note that $\hat{x}_1$ and $\hat{x}_2$ are orthogonal, that is, using Equation \eqref{xx}, we have
\begin{equation}
\begin{aligned}
&\hat{x}_1^{\mathrm{T}}\hat{x}_2=q-1+x_1 x_2=0.
\end{aligned}
\end{equation}
We define the normalized eigenvectors $\hat{\beta}_1$ and $\hat{\beta}_2$ by
\begin{align}
\hat{\beta}_1 \equiv \dfrac{\hat{x}_1}{||\hat{x}_1||},\quad \hat{\beta}_2 \equiv \dfrac{\hat{x}_2}{||\hat{x}_2||}.
\end{align}
We now define the vectors $\{\hat{y}_i\}_{i=1}^{q-2}$ by the following:
\begin{equation}
y_{i, k}=\begin{cases}0&, k>i+1\\1&, k<i+1\\-i&, k=i+1 \end{cases},
\end{equation}
which in matrix notation are given by
\begin{equation}
\begin{aligned}\hat{y}_1=\begin{pmatrix}1\\-1\\0\\\vdots\\0\end{pmatrix}, \hat{y}_2=\begin{pmatrix}1\\1\\-2\\0\\\vdots\\0\end{pmatrix}, \hat{y}_3=\begin{pmatrix}1\\1\\1\\-3\\0\\\vdots\\0\end{pmatrix},\dots, \hat{y}_{q-2}=\begin{pmatrix}1\\1\\\vdots\\1\\-(q-2)\\0\end{pmatrix}.
\end{aligned}
\end{equation}
The length of $\hat{y}_i$ is given by
\begin{equation}
||\hat{y}_i||=\sqrt{i(i+1)}.
\end{equation}
One can check that $\hat{y}_i$ is an eigenvector of $\hat{C}^{(q)}$ with eigenvalue 0. The normalized eigenvectors $\hat{\alpha}_i$ are defined by
\begin{equation}
\hat{\alpha}_i\equiv \dfrac{\hat{y}_i}{||\hat{y}_i||}.
\end{equation}
It can be checked that the set $\{\hat{\alpha}_i\}_{i=1}^{q-2}\cup \{\hat{\beta}_1, \hat{\beta}_2\}$ are orthonormal. Therefore, we can form the orthogonal matrix $\hat{\Lambda}^{(q)}$ with its columns being eigenvectors of $\hat{C}^{(q)}$, that is,
\begin{equation}
\hat{\Lambda}^{(q)}=(\hat{\alpha}_1, \hat{\alpha}_2, \dots, \hat{\alpha}_{q-2}, \hat{\beta}_1, \hat{\beta}_2),
\end{equation}
and we can write
\begin{align}
\left(\hat{\Lambda}^{(q)}\right)^{\mathrm{T}}\hat{C}^{(q)}\hat{\Lambda}^{(q)}=\begin{pmatrix}
0 & 0 & \cdots & 0 &0\\
0 & 0 & \cdots & 0 &0\\
\vdots  & \vdots  & \ddots & \vdots&\vdots  \\
0 & 0 & \cdots & \lambda_1&0\\
0&0&\cdots&0&\lambda_2
\end{pmatrix}.
\end{align}
\section{Proof of Lemma \ref{uvarlemm}}
\label{uvarproof}
We perform the following change of variables:
\begin{equation}
\hat{\mathbf{z}}=\hat{\Lambda}^{(q)}\hat{\mathbf{u}}=\sum_{m=1}^{q-2}  \hat{\alpha}_m \otimes \mathbf{u}_m+ \hat{\beta}_{1}\otimes \mathbf{u}_{q-1} + \hat{\beta}_{2} \otimes \mathbf{u}_{q},
\end{equation}
and we can write
\begin{equation}
\mathbf{z}_l=\sum_{m=1}^{q-2} {\alpha}_{m, l} \mathbf{u}_m+ {\beta}_{1, l} \mathbf{u}_{q-1} + {\beta}_{2, l}  \mathbf{u}_{q}.
\end{equation}
If we take the dot product, we obtain
\begin{equation}
\begin{aligned}
&\mathbf{z}_l \cdot \mathbf{z}_l=\sum_{m=1}^{q-2}\sum_{m'=1}^{q-2} {\alpha}_{m, l}{\alpha}_{m', l} \mathbf{u}_m \cdot \mathbf{u}_{m'}+2\beta_{1, l}\sum_{m=1}^{q-2} {\alpha}_{m, l} \mathbf{u}_m \cdot \mathbf{u}_{q-1}+2\beta_{2, l}\sum_{m=1}^{q-2}  \alpha_{m, l}\mathbf{u}_m \cdot \mathbf{u}_{q}\\
&+\beta_{1, l}^2 \mathbf{u}_{q-1} \cdot \mathbf{u}_{q-1}+2\beta_{1, l}\beta_{2, l} \mathbf{u}_{q-1} \cdot \mathbf{u}_{q}+\beta_{2, l}^2 \mathbf{u}_{q} \cdot \mathbf{u}_{q}.
\end{aligned}
\end{equation}
Note that we can write
\begin{equation}
\beta_{1, l}=\dfrac{1}{||\hat{x}_1||}\left[1+\delta_{q, l}\left(x_1-1\right)\right],\quad \beta_{2, l}=\dfrac{1}{||\hat{x}_2||}\left[1+\delta_{q, l}\left(x_2-1\right)\right],
\end{equation}
which gives
\begin{equation}
\begin{aligned}
&\mathbf{z}_l \cdot \mathbf{z}_l=\sum_{m=1}^{q-2}\sum_{m'=1}^{q-2} {\alpha}_{m, l}{\alpha}_{m', l} \mathbf{u}_m \cdot \mathbf{u}_{m'}+\dfrac{2}{||\hat{x}_1||}\sum_{m=1}^{q-2} {\alpha}_{m, l} \mathbf{u}_m \cdot \mathbf{u}_{q-1}+\dfrac{2}{||\hat{x}_2||}\sum_{m=1}^{q-2}  \alpha_{m, l}\mathbf{u}_m \cdot \mathbf{u}_{q}\\
&+\dfrac{1}{||\hat{x}_1||^2} \mathbf{u}_{q-1} \cdot \mathbf{u}_{q-1}+\dfrac{2}{||\hat{x}_1||||\hat{x}_2||} \mathbf{u}_{q-1} \cdot \mathbf{u}_{q}+\dfrac{1}{||\hat{x}_2||^2} \mathbf{u}_{q} \cdot \mathbf{u}_{q}\\
&+2\delta_{q, l}\dfrac{(x_2-1)}{||\hat{x}_2||}\left(\sum_{m=1}^{q-2}  \alpha_{m, l}\mathbf{u}_m \cdot \mathbf{u}_{q}+\dfrac{\mathbf{u}_{q-1} \cdot \mathbf{u}_{q}}{||\hat{x}_1||}+\dfrac{x_2+1}{2||\hat{x}_2||}\mathbf{u}_{q} \cdot \mathbf{u}_{q}\right)\\
&+2\delta_{q, l}\dfrac{(x_1-1)}{||\hat{x}_1||}\left(\sum_{m=1}^{q-2}  \alpha_{m, l}\mathbf{u}_m \cdot \mathbf{u}_{q-1}+\dfrac{\mathbf{u}_{q-1} \cdot \mathbf{u}_{q}}{||\hat{x}_2||}+\dfrac{x_1+1}{2||\hat{x}_1||}\mathbf{u}_{q-1} \cdot \mathbf{u}_{q-1}\right)\\
&+2\delta_{q, l}\dfrac{(x_1-1)(x_2-1)}{||\hat{x}_1||||\hat{x}_2||}\mathbf{u}_{q-1} \cdot \mathbf{u}_{q},
\end{aligned}
\end{equation}
\noindent and note that $\alpha_{m,q}=0$, so we have
\begin{equation}
\begin{aligned}
&\mathbf{z}_l \cdot \mathbf{z}_l=\sum_{m=1}^{q-2}\sum_{m'=1}^{q-2} {\alpha}_{m, l}{\alpha}_{m', l} \mathbf{u}_m \cdot \mathbf{u}_{m'}+\dfrac{2}{||\hat{x}_1||}\sum_{m=1}^{q-2} {\alpha}_{m, l} \mathbf{u}_m \cdot \mathbf{u}_{q-1}+\dfrac{2}{||\hat{x}_2||}\sum_{m=1}^{q-2}  \alpha_{m, l}\mathbf{u}_m \cdot \mathbf{u}_{q}\\
&+\dfrac{1}{||\hat{x}_1||^2} \mathbf{u}_{q-1} \cdot \mathbf{u}_{q-1}+\dfrac{2}{||\hat{x}_1||||\hat{x}_2||} \mathbf{u}_{q-1} \cdot \mathbf{u}_{q}+\dfrac{1}{||\hat{x}_2||^2} \mathbf{u}_{q} \cdot \mathbf{u}_{q}\\
&+2\delta_{q, l}\dfrac{(x_2-1)}{||\hat{x}_2||}\left(\dfrac{\mathbf{u}_{q-1} \cdot \mathbf{u}_{q}}{||\hat{x}_1||}+\dfrac{x_2+1}{2||\hat{x}_2||}\mathbf{u}_{q} \cdot \mathbf{u}_{q}\right)\\
&+2\delta_{q, l}\dfrac{(x_1-1)}{||\hat{x}_1||}\left(\dfrac{\mathbf{u}_{q-1} \cdot \mathbf{u}_{q}}{||\hat{x}_2||}+\dfrac{x_1+1}{2||\hat{x}_1||}\mathbf{u}_{q-1} \cdot \mathbf{u}_{q-1}\right)\\
&+2\delta_{q, l}\dfrac{(x_1-1)(x_2-1)}{||\hat{x}_1||||\hat{x}_2||}\mathbf{u}_{q-1} \cdot \mathbf{u}_{q}\\
&=G_l+2\delta_{q, l}T+\dfrac{1}{||\hat{x}_1||^2} \mathbf{u}_{q-1} \cdot \mathbf{u}_{q-1}+\dfrac{2}{||\hat{x}_1||||\hat{x}_2||} \mathbf{u}_{q-1} \cdot \mathbf{u}_{q}+\dfrac{1}{||\hat{x}_2||^2} \mathbf{u}_{q} \cdot \mathbf{u}_{q},
\end{aligned}
\end{equation}
\noindent  where we define
\begin{align}
\label{Geqapp}
&G_l\equiv \sum_{m=1}^{q-2}\sum_{m'=1}^{q-2} {\alpha}_{m, l}{\alpha}_{m', l} \mathbf{u}_m \cdot \mathbf{u}_{m'}+\dfrac{2}{||\hat{x}_1||}\sum_{m=1}^{q-2} {\alpha}_{m, l} \mathbf{u}_m \cdot \mathbf{u}_{q-1}+\dfrac{2}{||\hat{x}_2||}\sum_{m=1}^{q-2}  \alpha_{m, l}\mathbf{u}_m \cdot \mathbf{u}_{q},\\
\nonumber
\label{Teqapp}
&T\equiv \dfrac{(x_2-1)}{||\hat{x}_2||}\left(\dfrac{\mathbf{u}_{q-1} \cdot \mathbf{u}_{q}}{||\hat{x}_1||}+\dfrac{x_2+1}{2||\hat{x}_2||}\mathbf{u}_{q} \cdot \mathbf{u}_{q}\right)+\dfrac{(x_1-1)}{||\hat{x}_1||}\left(\dfrac{\mathbf{u}_{q-1} \cdot \mathbf{u}_{q}}{||\hat{x}_2||}+\dfrac{x_1+1}{2||\hat{x}_1||}\mathbf{u}_{q-1} \cdot \mathbf{u}_{q-1}\right)\\
&+\dfrac{(x_1-1)(x_2-1)}{||\hat{x}_1||||\hat{x}_2||}\mathbf{u}_{q-1} \cdot \mathbf{u}_{q},
\end{align}
\noindent  where note that $G_l=0$ if $l=q$, and so we have
\begin{equation}
\label{0cond}
G_l+2\delta_{l, q}T=\begin{cases}G_l&, l\neq q\\
2T&, l=q\end{cases}.
\end{equation}
Rearranging Equation \eqref{Teqapp} further yields the desired expression for $T$. Writing $F$ in terms of the new variables, we obtain
\begin{equation}
\label{fnewvar}
\begin{aligned}
&F(\mathbf{z}_1,\dots,\mathbf{z}_q)=-\dfrac{\mu}{2}\left(\dfrac{1}{||\hat{x}_1||^2} \mathbf{u}_{q-1} \cdot \mathbf{u}_{q-1}+\dfrac{2}{||\hat{x}_1||||\hat{x}_2||} \mathbf{u}_{q-1} \cdot \mathbf{u}_{q}+\dfrac{1}{||\hat{x}_2||^2} \mathbf{u}_{q} \cdot \mathbf{u}_{q}\right)\\
&+\ln \dfrac{1}{q}\sum_l \exp \left\{-\dfrac{\mu}{2}\left(G_l+2\delta_{q, l}T\right)\right\}.
\end{aligned}
\end{equation}
\section{Proof of Theorem \ref{entrocaltheo}}
\label{entapp}
We have $\mathrm{d}\mathbf{z}_1 \cdots \mathrm{d}\mathbf{z}_q=\mathrm{d}\mathbf{u}_1 \cdots \mathrm{d}\mathbf{u}_q$ since $\hat{\Lambda}^{(q)}$ is orthogonal, and we also have
\begin{equation}
\begin{aligned}
\exp\left\{-\dfrac{1}{2}\mathbf{\hat{z}}^{\mathrm{T}}\hat{M}^{(q)} \mathbf{\hat{z}}\right\}= \exp\left\{-\dfrac{1}{2}\mathbf{\hat{u}}^{\mathrm{T}}(\hat{\Lambda}^{(q)})^{\mathrm{T}}\hat{M}^{(q)} \hat{\Lambda}^{(q)}\mathbf{\hat{u}}\right\}=\exp\left\{-\dfrac{1}{2}\mathbf{\hat{u}}^{\mathrm{T}}\hat{D}^{(q)} \mathbf{\hat{u}}\right\}.
\end{aligned}
\end{equation}
It then follows directly from Lemma \ref{uvarlemm} that
\begin{equation}
\label{expec1}
\begin{aligned}
&\int \exp\left\{-\dfrac{1}{2}\mathbf{\hat{z}}^{\mathrm{T}}\hat{M}^{(q)} \mathbf{\hat{z}}\right\} F(\mathbf{z}_1,\dots,\mathbf{z}_q)\mathrm{d}\mathbf{z}_1 \cdots \mathrm{d}\mathbf{z}_q=-\dfrac{\mu}{2}\int \exp\left\{-\dfrac{1}{2}\mathbf{\hat{u}}^{\mathrm{T}}\hat{D}^{(q)} \mathbf{\hat{u}}\right\} \cdot\\ &\left(\dfrac{1}{||\hat{x}_1||^2} \mathbf{u}_{q-1} \cdot \mathbf{u}_{q-1}+\dfrac{1}{||\hat{x}_2||^2} \mathbf{u}_{q} \cdot \mathbf{u}_{q}+\dfrac{2}{||\hat{x}_1||||\hat{x}_2||} \mathbf{u}_{q-1} \cdot \mathbf{u}_{q}\right)\mathrm{d}\mathbf{u}_1 \cdots \mathrm{d}\mathbf{u}_q\\
&+\int  \exp\left\{-\dfrac{1}{2}\mathbf{\hat{u}}^{\mathrm{T}}\hat{D}^{(q)} \mathbf{\hat{u}}\right\} \ln q^{-1}\sum_l \exp \left\{-\dfrac{\mu}{2}\left(G_l+2\delta_{q, l}T\right)\right\}\mathrm{d}\mathbf{u}_1 \cdots \mathrm{d}\mathbf{u}_q.
\end{aligned}
\end{equation}
Note that we have
\begin{equation}
\begin{aligned}
&\mathbf{\hat{u}}^{\mathrm{T}}\hat{D}^{(q)} \mathbf{\hat{u}}=\sum_{k=1}^{q-2} \mathbf{u}_k \cdot \mathbf{u}_k+m_1\mathbf{u}_{q-1} \cdot \mathbf{u}_{q-1}+m_2\mathbf{u}_{q} \cdot \mathbf{u}_{q}\\
&=\sum_{k=1}^{q-2} \sum_{i=1}^{n}u^2_{k, i}+m_1\sum_{i=1}^{n}u^2_{q-1, i}+m_2\sum_{i=1}^{n}u^2_{q, i},
\end{aligned}
\end{equation}
and so we have
\begin{equation}
\label{someq}
\exp\left\{-\dfrac{1}{2}\mathbf{\hat{u}}^{\mathrm{T}}\hat{D}^{(q)} \mathbf{\hat{u}}\right\}=\prod_{i=1}^{n}\prod_{k=1}^{q}\exp\left\{ -\dfrac{1}{2}d_{k}u_{k, i}^2\right\},
\end{equation}
where we define
\begin{equation}
d_{k}\equiv \begin{cases}1&, k \leq q-2\\
m_1&, k=q-1\\
m_2&, k=q\end{cases}.
\end{equation}
Rearranging Equation \eqref{someq}, we have
\begin{equation}
\begin{aligned}
\exp\left\{-\dfrac{1}{2}\mathbf{\hat{u}}^{\mathrm{T}}\hat{D}^{(q)} \mathbf{\hat{u}}\right\}=\prod_{i=1}^{n}\prod_{k=1}^{q}\dfrac{\sqrt{2\pi}}{\sqrt{d_{k}}}\dfrac{\sqrt{d_{k}}}{\sqrt{2\pi}}\exp\left\{ -\dfrac{1}{2}\left(\dfrac{u_{k}}{(1/\sqrt{d_{k}})}\right)^2\right\}=\prod_{i=1}^{n}\prod_{k=1}^{q}\dfrac{\sqrt{2\pi}}{\sqrt{d_{k}}}g_{0, d^{-1}_{k}}(u_{k, i}),
\end{aligned}
\end{equation}
\normalsize
where $g_{\tilde{\mu}, \sigma^2}(x)$ is the Gaussian function with mean $\tilde{\mu}$ and variance $\sigma^2$.
The first term in Equation \eqref{expec1} yields
\begin{equation}
\label{secterm}
\begin{aligned}
&-\dfrac{\mu}{2||\hat{x}_1||^{2}}\int \exp\left\{-\dfrac{1}{2}\mathbf{\hat{u}}^{\mathrm{T}}\hat{D}^{(q)} \mathbf{\hat{u}}\right\} \mathbf{u}_{q-1}\cdot \mathbf{u}_{q-1}\mathrm{d}\mathbf{u}_1 \cdots \mathrm{d}\mathbf{u}_q\\
&=-\dfrac{\mu}{2||\hat{x}_1||^{2}}\sum_{j=1}^{n}\int \exp\left\{-\dfrac{1}{2}\mathbf{\hat{u}}^{\mathrm{T}}\hat{D}^{(q)} \mathbf{\hat{u}}\right\} u^2_{q-1,j}\mathrm{d}\mathbf{u}_1 \cdots \mathrm{d}\mathbf{u}_q\\
&=-\dfrac{\mu}{2||\hat{x}_1||^{2}}\sum_{j=1}^{n}\int u^2_{q-1,j}\prod_{i=1}^{n}\prod_{k=1}^{q}\dfrac{\sqrt{2\pi}}{\sqrt{d_{k}}}g_{0, d^{-1}_{k}}(u_{k, i})\mathrm{d}u_{k, i}=-\dfrac{\mu(\sqrt{2\pi})^{nq}}{2||\hat{x}_1||^{2}} \times\\
&\sum_{j=1}^{n}\left(\int \prod_{(k,i)\neq(q-1,j)}\dfrac{1}{\sqrt{d_{k}}}g_{0, d^{-1}_{k}}(u_{k, i})\mathrm{d}u_{k, i}\right)\left(\dfrac{1}{\sqrt{d_{q-1}}}\int g_{0, d^{-1}_{q-1}}(u_{q-1, j})u^2_{q-1,j}\mathrm{d}u_{q-1, j}\right)\\
&=-\dfrac{\mu(\sqrt{2\pi})^{nq}}{2||\hat{x}_1||^{2}}\sum_{j=1}^{n}\left(\prod_{(k,i)\neq(q-1,j)}\dfrac{1}{\sqrt{d_{k}}}\right)\left(\dfrac{d^{-1}_{q-1}}{\sqrt{d_{q-1}}}\right)\\
&=-\dfrac{\mu(\sqrt{2\pi})^{nq}}{2||\hat{x}_1||^{2}}\sum_{j=1}^{n}\left(\prod_{i=1}^{n}\prod_{k=1}^{q}\dfrac{1}{\sqrt{d_{k}}}\right)d^{-1}_{q-1}
=-\dfrac{\mu(\sqrt{2\pi})^{nq}}{2||\hat{x}_1||^{2}}\sum_{j=1}^{n}\dfrac{1}{\left(\sqrt{m_1 m_2}\right)^n}\dfrac{1}{m_1}\\
&=-\dfrac{n\mu(\sqrt{2\pi})^{nq}}{2||\hat{x}_1||^{2}}\dfrac{1}{\left(\sqrt{(m_1 m_2)}\right)^n}\dfrac{1}{m_1}.
\end{aligned}
\end{equation}
In an analogous calculation, the second term in Equation \eqref{expec1} yields
\begin{equation}
\label{thirdterm}
\begin{aligned}
-\dfrac{\mu}{2||\hat{x}_2||^{2}}\int \exp\left\{-\dfrac{1}{2}\mathbf{\hat{u}}^{\mathrm{T}}\hat{D}^{(q)} \mathbf{\hat{u}}\right\} \mathbf{u}_{q}\cdot \mathbf{u}_{q}\mathrm{d}\mathbf{u}_1 \cdots \mathrm{d}\mathbf{u}_q=-\dfrac{n\mu(\sqrt{2\pi})^{nq}}{2||\hat{x}_2||^{2}}\dfrac{1}{\left(\sqrt{m_1 m_2}\right)^n}\dfrac{1}{m_2}.
\end{aligned}
\end{equation}
The third term in Equation \eqref{expec1} is identically zero since the matrix $\hat{D}^{(q)}$ is diagonal (there is no covariance). From Corollary \ref{entropcoll}, Equations \eqref{simpli},  \eqref{simpli2}, \eqref{secterm} and \eqref{thirdterm},  we have
\begin{equation}
\label{entrop1}
\begin{aligned}
&h(\vecX|\hat\vecW)=n \ln (\sigma\sqrt{2\pi})+\dfrac{n}{2}\left(1+\mu\dfrac{q}{q-1}\right)\\
&-\left(\dfrac{\mu}{q}\right)^{n/2}(\sqrt{2\pi})^{-qn}\int \exp\left\{-\dfrac{1}{2}\mathbf{\hat{u}}^{\mathrm{T}}\hat{D}^{(q)} \mathbf{\hat{u}}\right\}  \ln q^{-1}\sum_l \exp \left\{-\dfrac{\mu}{2}\left(G_l+2\delta_{q, l}T\right)\right\}\mathrm{d}\mathbf{u}_1 \cdots \mathrm{d}\mathbf{u}_q,\\
&=n h_\sigma+\dfrac{n}{2}\dfrac{q}{q-1}\mu-\mathbb{E}_{u}\left[  \ln  q^{-1}\sum_l \exp \left\{-\dfrac{\mu}{2}\left(G_l+2\delta_{q, l}T\right)\right\}\right],
\end{aligned}
\end{equation}
where $\mathbb{E}_u$ is the expectation taken with respect to the density $\prod_{i=1}^{n}\prod_{k=1}^{q}g_{0, d^{-1}_{k}}(u_{k, i})$.
\section{Proof of Lemma \ref{wlemm}}
\label{tildproof}
We have $\mathbf{u}_q = \mathbf{r}/\sqrt{m_2}$, $ \mathbf{u}_{q-1} = \mathbf{s}/\sqrt{m_1}$. This normalizes the variance of $\vecr,\vecs$ to~1. Hence, each component $r_1,\ldots,r_n$, $s_1,\ldots,s_n$ is an independent normal-distributed variable. We show that it is consistent to write ${\binom{\vecr}{\vecs}} = W{\binom{\tilde\vecr}{\tilde\vecs}}$, with $W$ as given in (\ref{wmat}), and all the components $\tilde r_1,\ldots,\tilde r_n$, $\tilde s_1,\ldots,\tilde s_n$ are independent and normally distributed. 
The change of variables can be written component-wise as $r_\alpha = W_{11}\tilde r_\alpha+W_{12}\tilde s_\alpha$, $s_\alpha = W_{21}\tilde r_\alpha+W_{22}\tilde s_\alpha$, for $\alpha\in\{1,\ldots,n\}$. A linear combination of normal-distributed variables has a Gaussian distribution. Furthermore, we have first-order statistics $\EE [r_\alpha] =0$, $\EE [s_\alpha] =0$. The second-order statistics are given by
\begin{eqnarray}
\EE [r_\alpha^2] &=&
W_{11}^2 \EE [\tilde r_\alpha^2] + W_{12}^2 \EE [\tilde s_\alpha^2] =
\frac\mu{1+\mu\frac{q}{q-1}}(  \frac{1/m_2}{q-1+x_2^2} + \frac{1/m_1}{q-1+x_1^2}  )
\stackrel{(\ref{simpli2})}{=} 1,
\\
\EE [s_\alpha^2] &=&
W_{21}^2 \EE [\tilde r_\alpha^2] + W_{22}^2 \EE [\tilde s_\alpha^2] =
\frac\mu{1+\mu\frac{q}{q-1}}(  \frac{1/m_1}{q-1+x_1^2} + \frac{1/m_2}{q-1+x_2^2}  )
\stackrel{(\ref{simpli2})}{=} 1,
\\
\EE [r_\alpha s_\beta] &=&
W_{11} W_{21}\EE [\tilde r_\alpha \tilde r_\beta] + W_{12} W_{22}\EE [\tilde s_\alpha \tilde s_\beta]
= \delta_{\alpha\beta}(W_{11} W_{21} + W_{12} W_{22})=0,
\end{eqnarray}
and $\EE [r_\alpha r_\beta]=\delta_{\alpha\beta}$, $\EE [s_\alpha s_\beta]=\delta_{\alpha\beta}$. The integration measure changes as \begin{equation} {\rm d}r_\alpha {\rm d}s_\alpha = |\det W|  {\rm d}\tilde r_\alpha {\rm d}\tilde s_\alpha,\end{equation} with
\begin{equation}
\det W =
\frac\mu{1+\mu\frac{q}{q-1}}(  -\frac{1/m_1}{q-1+x_1^2} - \frac{1/m_2}{q-1+x_2^2}  )
\stackrel{(\ref{simpli2})}{=} -1.
\end{equation}
We replace the expectation $\EE_{u}$ by $\EE_{u, \tilde{s}, \tilde{r}}$.
\section{Proof of Lemma \ref{gausdetlemm}}
\label{gausdetproof}
\begin{equation}
\label{deteq}
\begin{aligned}
&\int \exp\left\{-\dfrac{1}{2}\mathbf{\hat{u}}^{\mathrm{T}}\hat{V} \mathbf{\hat{u}}\right\}  \mathrm{d}\mathbf{u}_1 \cdots \mathrm{d}\mathbf{u}_q=\int \prod_{k=1}^{n}\exp\left\{-\dfrac{1}{2}\hat{u}^{\mathrm{T}}_{k}\hat{V} \hat{u}_{k}\right\} \mathrm{d}\hat{u}_{k}=\left(\int \exp\left\{-\dfrac{1}{2}\hat{u}^{\mathrm{T}}\hat{V} \hat{u}\right\} \mathrm{d}\hat{u}\right)^n\\
&=\left(\sqrt{(2\pi)^q\det(\hat{V}^{-1})}\right)^n\left(\left(\sqrt{(2\pi)^q\det(\hat{V}^{-1})}\right)^{-1}\int \exp\left\{-\dfrac{1}{2}\hat{u}^{\mathrm{T}}\hat{V} \hat{u}\right\} \mathrm{d}\hat{u}\right)^n=\dfrac{(2\pi)^{nq/2}}{\det(\hat{V})^{n/2}}.
\end{aligned}
\end{equation}
\normalsize
\section{Proof of Lemma \ref{ssereieslemm}}
\label{Zcalcapp}
\begin{equation}
\begin{aligned}
\left(\sqrt{2\pi}\right)^{-nq}\int \exp\left\{-\dfrac{1}{2}{\hat{\mathbf{a}}}^{\mathrm{T}}\hat{I} \hat{\mathbf{a}}\right\}  Z\mathrm{d}{\mathbf{a}}_1 \cdots \mathrm{d}{\mathbf{a}}_q=\dfrac{1}{\det(\hat{I})^{n/2}}-\sum_{l=1}^{q-1}\dfrac{q^{-1}}{\det(\hat{I}+\hat{Q}_l)^{n/2}}-\dfrac{q^{-1}}{\det(\hat{I}+\hat{P})^{n/2}}.
\end{aligned}
\end{equation}
\normalsize
\noindent For the $Z^2$ contribution, we have
\begin{equation}
\begin{aligned}
&Z^2=Z-q^{-1}\sum_{l=1}^{q-1} \exp \left\{-\dfrac{1}{2}\hat{\mathbf{a}}^\mathrm{T} \hat{Q}_l \hat{\mathbf{a}}\right\}-q^{-1}\exp \left\{-\dfrac{1}{2}\hat{\mathbf{a}}^\mathrm{T} \hat{P} \hat{\mathbf{a}}\right\}+2q^{-2}\sum_{l=1}^{q-1} \exp \left\{-\dfrac{1}{2}\hat{\mathbf{a}}^\mathrm{T} (\hat{Q}_l+\hat{P}) \hat{\mathbf{a}}\right\}\\
&+q^{-2}\sum_{l=1}^{q-1}\sum_{l'=1}^{q-1}\exp \left\{-\dfrac{1}{2}\hat{\mathbf{a}}^\mathrm{T} (\hat{Q}_l+\hat{Q}_{l'}) \hat{\mathbf{a}}\right\}+q^{-2}\exp \left\{-\dfrac{1}{2}\hat{\mathbf{a}}^\mathrm{T} 2\hat{P} \hat{\mathbf{a}}\right\},
\end{aligned}
\end{equation}
\normalsize
and so we have
\begin{equation}
\begin{aligned}
&\dfrac{1}{2}\left(\sqrt{2\pi}\right)^{-nq}\int \exp\left\{-\dfrac{1}{2}{\hat{\mathbf{a}}}^{\mathrm{T}}\hat{I} \hat{\mathbf{a}}\right\}  Z^2\mathrm{d}{\mathbf{a}}_1 \cdots \mathrm{d}{\mathbf{a}}_q=\dfrac{(1/2)}{\det(\hat{I})^{n/2}}-\sum_{l=1}^{q-1}\dfrac{q^{-1}}{\det(\hat{I}+\hat{Q}_l)^{n/2}}-\dfrac{q^{-1}}{\det(\hat{I}+\hat{P})^{n/2}}\\
&+\sum_{l=1}^{q-1}\dfrac{q^{-2}}{\det(\hat{I}+\hat{Q}_l+\hat{P})^{n/2}}+\sum_{l=1}^{q-1}\sum_{l'=1}^{q-1} \dfrac{(1/2)q^{-2}}{\det(\hat{I}+\hat{Q}_l+\hat{Q}_{l'})^{n/2}}+\dfrac{(1/2)q^{-2}}{\det(\hat{I}+2\hat{P})^{n/2}}.
\end{aligned}
\end{equation}
\normalsize
\section{Proof of Theorem \ref{bruttheo}}
\label{bruteforceproof}
For the purpose of this calculation, we rename our expressions for simplicity as the following:
\begin{equation}
\begin{aligned}
&aa \leftarrow \tilde{\mathbf{r}}\cdot \tilde{\mathbf{r}}, bb \leftarrow \tilde{\mathbf{s}}\cdot \tilde{\mathbf{s}}, ab\leftarrow\tilde{\mathbf{r}}\cdot \tilde{\mathbf{s}}, A \leftarrow  \sum_{l=1}^{q-1}\mathbf{w}_l \cdot \mathbf{w}_l, B \leftarrow \sum_{l=1}^{q-1}(\mathbf{w}_l \cdot \tilde{\mathbf{r}})^2, D \leftarrow \sum_{l=1}^{q-1}(\mathbf{w}_l \cdot \mathbf{w}_l)(\mathbf{w}_l \cdot \tilde{\mathbf{r}}),\\
&F \leftarrow \sum_{l=1}^{q-1}(\mathbf{w}_l \cdot \mathbf{w}_l)^2,
H \leftarrow \sum_{l=1}^{q-1}(\mathbf{w}_l \cdot \mathbf{w}_l)(\mathbf{w}_l \cdot \tilde{\mathbf{r}})^2,
M \leftarrow \sum_{l=1}^{q-1}(\mathbf{w}_l \cdot \tilde{\mathbf{r}})^3,
J \leftarrow \sum_{l=1}^{q-1}(\mathbf{w}_l \cdot \tilde{\mathbf{r}})^4.
\end{aligned}
\end{equation}
\normalsize
Using Equation \eqref{Sdef} and Lemma \ref{wlemm}, we write the Taylor series for $S$ as the following:
\begin{align}
&\ln q^{-1}\sum_l \exp \left\{-\dfrac{\mu}{2}\left(G_l+2\delta_{q, l}T\right)\right\}=c_{1/2}\sqrt{\mu}+c_{1}\mu+c_{3/2}\mu^{3/2}+c_{2}\mu^2+\mathcal{O}\left(\mu^{5/2}\right),\\
&S=-\mathbb{E}_{u,\tilde{s},\tilde{r}}[c_{1/2}]\sqrt{\mu}-\mathbb{E}_{u,\tilde{s},\tilde{r}}[c_{1}]\mu-\mathbb{E}_{u,\tilde{s},\tilde{r}}[c_{3/2}]\mu^{3/2}-\mathbb{E}_{u,\tilde{s},\tilde{r}}[c_{2}]\mu^2+\mathcal{O}\left(\mu^{3}\right),
\end{align}
where we use Mathematica to obtain the following expressions for $c_1$ and $c_2$:
\begin{align}
&2q(q-1)c_1=-(q-1)A+(q-1)B+2qaa-qbb+(q-1)ab^2,\\
\nonumber
&24 q^2 (q-1)^2 c_2=-3(q-1)^2A^2-3B^2+12 q aa B+ 6qB^2-6q bb B+3q F-6qH+qJ-12q^2 aa^2\\
\nonumber
&-12 q^2 B -12 q^2 aa B
-3 q^2 B^2+12 q^2 aa bb + 6 q^2 bbB-3q^2 bb^2-6q^2F+12q^2H-2q^2 J-12 q^3 aa\\
\nonumber
&+12 q^3 aa^2+12q^3 B+12q^3 bb-12 q^3 aa bb+3q^3 bb^2+3q^3 F-6q^3 H+q^3 J\\
\nonumber
&+6(q-1)[(q-2)ab^2 A+(q-1)AB+2q aa A-q bb A]+(q^3-7q^2+12q-6)ab^4\\
&-6(q-1)[(q-2)ab^2 B-2q(q-2)ab^2 aa+q(q-2)ab^2 bb+4q^2 ab^2].
\end{align}
\normalsize
Furthermore, the expression for $c_{1/2}$ and $c_{3/2}$ do not contribute to $S$ since they vanish under the expectation. We now analytically compute  all the relevant terms needed.
\begin{equation}
\begin{aligned}
\mathbb{E}_{\{u, \tilde{s}, \tilde{r}\}}\left[aa\right]=\mathbb{E}_{\{\tilde{r}\}}\left[\tilde{\mathbf{r}} \cdot \tilde{\mathbf{r}}\right]=n,\quad \mathbb{E}_{\{u, \tilde{s}, \tilde{r}\}}\left[bb\right]=\mathbb{E}_{\{\tilde{s}\}}\left[\tilde{\mathbf{s}} \cdot \tilde{\mathbf{s}}\right]=n,
\end{aligned}
\end{equation}
\begin{equation}
\begin{aligned}
\mathbb{E}_{\{u, \tilde{s}, \tilde{r}\}}\left[aa^2\right]=\mathbb{E}_{\{\tilde{r}\}}\left[(\tilde{\mathbf{r}} \cdot \tilde{\mathbf{r}})^2\right]=n(n+2).
\end{aligned}
\end{equation}
\begin{equation}
\begin{aligned}
\mathbb{E}_{\{u, \tilde{s}, \tilde{r}\}}\left[bb^2\right]=\mathbb{E}_{\{\tilde{s}\}}\left[(\tilde{\mathbf{s}} \cdot \tilde{\mathbf{s}})^2\right]=n(n+2).
\end{aligned}
\end{equation}
\begin{equation}
\begin{aligned}
\mathbb{E}_{\{u, \tilde{s}, \tilde{r}\}}\left[aabb\right]=\mathbb{E}_{\{\tilde{s}, \tilde{r}\}}\left[(\tilde{\mathbf{r}} \cdot \tilde{\mathbf{r}})(\tilde{\mathbf{s}}\cdot \tilde{\mathbf{s}})\right]=n^2.
\end{aligned}
\end{equation}
\begin{equation}
\begin{aligned}
\mathbb{E}_{\{u, \tilde{s}, \tilde{r}\}}\left[ab^2\right]=\mathbb{E}_{\{\tilde{s}, \tilde{r}\}}\left[(\tilde{\mathbf{r}}\cdot \tilde{\mathbf{s}})^2\right]=n.
\end{aligned}
\end{equation}
\begin{equation}
\begin{aligned}
&\mathbb{E}_{\{u, \tilde{s}, \tilde{r}\}}\left[ab^2 aa\right]=\mathbb{E}_{\{\tilde{s}, \tilde{r}\}}\left[(\tilde{\mathbf{r}}\cdot \tilde{\mathbf{s}})^2(\tilde{\mathbf{r}}\cdot \tilde{\mathbf{r}})\right]=\sum_{i=1}^{n}\sum_{j=1}^{n}\sum_{k=1}^{n}\mathbb{E}_{\{\tilde{s}, \tilde{r}\}}\left[\tilde{r}_{i}\tilde{r}_{j}\tilde{r}_{k}\tilde{r}_{k}\tilde{s}_{i}\tilde{s}_{j}\right]\\
&=\sum_{i=1}^{n}\sum_{j=1}^{n}\sum_{k=1}^{n}\mathbb{E}_{\{\tilde{r}\}}\left[\tilde{r}_{i}\tilde{r}_{j}\tilde{r}_{k}\tilde{r}_{k}\right]\mathbb{E}_{\{\tilde{s}\}}\left[\tilde{s}_{i}\tilde{s}_{j}\right]=\sum_{i=1}^{n}\sum_{j=1}^{n}\sum_{k=1}^{n}\mathbb{E}^{W}_{\{\tilde{r}\}}\left[\tilde{r}_{i}\tilde{r}_{j}\tilde{r}_{k}\tilde{r}_{k}\right]\delta_{i, j}=\mathbb{E}_{\{\tilde{r}\}}\left[(\tilde{\mathbf{r}} \cdot \tilde{\mathbf{r}})^2\right]=n(n+2).
\end{aligned}
\end{equation}

\begin{equation}
\begin{aligned}
&\mathbb{E}_{\{u, \tilde{s}, \tilde{r}\}}\left[ab^2 bb\right]=\mathbb{E}_{\{\tilde{s}, \tilde{r}\}}\left[(\tilde{\mathbf{r}}\cdot \tilde{\mathbf{s}})^2(\tilde{\mathbf{s}}\cdot \tilde{\mathbf{s}})\right]=\sum_{i=1}^{n}\sum_{j=1}^{n}\sum_{k=1}^{n}\mathbb{E}_{\{\tilde{s}, \tilde{r}\}}\left[\tilde{r}_{i}\tilde{r}_{j}\tilde{s}_{k}\tilde{s}_{k}\tilde{s}_{i}\tilde{s}_{j}\right]\\
&=\sum_{i=1}^{n}\sum_{j=1}^{n}\sum_{k=1}^{n}\mathbb{E}_{\{\tilde{r}\}}\left[\tilde{r}_{i}\tilde{r}_{j}\right]\mathbb{E}_{\{\tilde{s}\}}\left[\tilde{s}_{i}\tilde{s}_{j}\tilde{s}_{k}\tilde{s}_{k}\right]=\sum_{i=1}^{n}\sum_{j=1}^{n}\sum_{k=1}^{n}\mathbb{E}_{\{\tilde{s}\}}\left[\tilde{s}_{i}\tilde{s}_{j}\tilde{s}_{k}\tilde{s}_{k}\right]\delta_{i, j}=\mathbb{E}_{\{\tilde{s}\}}\left[(\tilde{\mathbf{s}} \cdot \tilde{\mathbf{s}})^2\right]=n(n+2).
\end{aligned}
\end{equation}
\small
\begin{equation}
\begin{aligned}
&\mathbb{E}_{\{u, \tilde{s}, \tilde{r}\}}\left[ab^4\right]=\mathbb{E}_{\{\tilde{s}, \tilde{r}\}}\left[(\tilde{\mathbf{r}}\cdot \tilde{\mathbf{s}})^4\right]=\sum_{i=1}^n\sum_{j=1}^n\sum_{k=1}^n\sum_{z=1}^n\mathbb{E}_{\{\tilde{s}, \tilde{r}\}}\left[\tilde{r}_{i}\tilde{r}_{j}\tilde{r}_{k}\tilde{r}_{z}\tilde{s}_{i}\tilde{s}_{j}\tilde{s}_{k}\tilde{s}_{z}\right]\\
&=\sum_{i=1}^n\sum_{j=1}^n\sum_{k=1}^n\sum_{z=1}^n\mathbb{E}_{\{\tilde{r}\}}\left[\tilde{r}_{i}\tilde{r}_{j}\tilde{r}_{k}\tilde{r}_{z}\right]\mathbb{E}_{\{\tilde{s}\}}\left[\tilde{s}_{i}\tilde{s}_{j}\tilde{s}_{k}\tilde{s}_{z}\right]\!\!=\!\!\sum_{i=1}^n\sum_{j=1}^n\sum_{k=1}^n\sum_{z=1}^n\mathbb{E}_{\{\tilde{r}\}}\left[\tilde{r}_{i}\tilde{r}_{j}\tilde{r}_{k}\tilde{r}_{z}\right]\left(\delta_{i, j}\delta_{k, z}+\delta_{i, k}\delta_{z, j}+\delta_{i, z}\delta_{k, j}\right)\\
&=\sum_{i=1}^n\sum_{k=1}^n\mathbb{E}_{\{\tilde{r}\}}\left[\tilde{r}_{i}\tilde{r}_{i}\tilde{r}_{k}\tilde{r}_{k}\right]+\sum_{i=1}^n\sum_{j=1}^n\mathbb{E}_{\{\tilde{r}\}}\left[\tilde{r}_{i}\tilde{r}_{j}\tilde{r}_{i}\tilde{r}_{j}\right]+\sum_{i=1}^n\sum_{j=1}^n\mathbb{E}_{\{\tilde{r}\}}\left[\tilde{r}_{i}\tilde{r}_{j}\tilde{r}_{j}\tilde{r}_{i}\right]=3\mathbb{E}_{\{\tilde{r}\}}\left[(\tilde{\mathbf{r}} \cdot \tilde{\mathbf{r}})^2\right]=3n(n+2).
\end{aligned}
\end{equation}
\normalsize
\begin{equation}
\begin{aligned}
&\mathbb{E}_{\{u, \tilde{s}, \tilde{r}\}}\left[A\right]=\sum_{l=1}^{q-1}\mathbb{E}_{\{u, \tilde{s}, \tilde{r}\}}\left[\mathbf{w}_l \cdot \mathbf{w}_l\right]=\sum_{l=1}^{q-1}\sum_{m=1}^{q-2}\sum_{m'=1}^{q-2}\alpha_{m, l}\alpha_{m', l}\mathbb{E}_{\{u, \tilde{s}, \tilde{r}\}}\left[\mathbf{u}_m \cdot \mathbf{u}_{m'}\right]\\
&=\sum_{l=1}^{q-1}\sum_{m=1}^{q-2}\sum_{m'=1}^{q-2}\alpha_{m, l}\alpha_{m', l}n \delta_{m, m'}=n\sum_{l=1}^{q-1} \sum_{m=1}^{q-2}\alpha^2_{m, l}=n  \sum_{m=1}^{q-2}\sum_{l=1}^{q}\alpha^2_{m, l}=n(q-2).
\end{aligned}
\end{equation}
\begin{equation}
\begin{aligned}
&\mathbb{E}_{\{u, \tilde{s}, \tilde{r}\}}\left[B\right]=\sum_{l=1}^{q-1}\mathbb{E}_{\{u, \tilde{s}, \tilde{r}\}}\left[(\mathbf{w}_l \cdot \tilde{\mathbf{r}})^2\right]=\sum_{l=1}^{q-1}\sum_{m=1}^{q-2}\sum_{m'=1}^{q-2}\alpha_{m, l}\alpha_{m', l}\mathbb{E}_{\{u, \tilde{s}, \tilde{r}\}}\left[(\mathbf{u}_m \cdot\tilde{\mathbf{r}}) (\mathbf{u}_{m'}\cdot \tilde{\mathbf{r}})\right]\\
&=\sum_{l=1}^{q-1}\sum_{m=1}^{q-2}\sum_{m'=1}^{q-2}\alpha_{m, l}\alpha_{m', l}n \delta_{m, m'}=n\sum_{l=1}^{q-1} \sum_{m=1}^{q-2}\alpha^2_{m, l}=n  \sum_{m=1}^{q-2}\sum_{l=1}^{q}\alpha^2_{m, l}=n(q-2).
\end{aligned}
\end{equation}
\begin{equation}
\begin{aligned}
&\mathbb{E}_{\{u, \tilde{s}, \tilde{r}\}}\left[A^2\right]=\mathbb{E}_{\{u, \tilde{s}, \tilde{r}\}}\left[\left(\sum_{l=1}^{q-1}\mathbf{w}_l \cdot \mathbf{w}_l\right)^2\right]\\
&=\sum_{l=1}^{q-1}\sum_{l'=1}^{q-1}\sum_{r=1}^{q-2}\sum_{r'=1}^{q-2}\sum_{m=1}^{q-2}\sum_{m'=1}^{q-2}\alpha_{m, l}\alpha_{m', l}\alpha_{r, l'}\alpha_{r', l'}\mathbb{E}_{\{u, \tilde{s}, \tilde{r}\}}\left[(\mathbf{u}_{m} \cdot \mathbf{u}_{m'})(\mathbf{u}_{r} \cdot \mathbf{u}_{r'})\right]\\
&=\sum_{r=1}^{q-2}\sum_{r'=1}^{q-2}\sum_{m=1}^{q-2}\sum_{m'=1}^{q-2}\left(\sum_{l=1}^{q-1}\alpha_{m, l}\alpha_{m', l}\right)\left(\sum_{l'=1}^{q-1}\alpha_{r, l'}\alpha_{r', l'}\right)\mathbb{E}_{\{u, \tilde{s}, \tilde{r}\}}\left[(\mathbf{u}_{m} \cdot \mathbf{u}_{m'})(\mathbf{u}_{r} \cdot \mathbf{u}_{r'})\right]\\
&=\sum_{r=1}^{q-2}\sum_{r'=1}^{q-2}\sum_{m=1}^{q-2}\sum_{m'=1}^{q-2}\delta_{m, m'}\delta_{r, r'}\mathbb{E}_{\{u, \tilde{s}, \tilde{r}\}}\left[(\mathbf{u}_{m} \cdot \mathbf{u}_{m'})(\mathbf{u}_{r} \cdot \mathbf{u}_{r'})\right]\\
&=\sum_{r=1}^{q-2}\sum_{m=1}^{q-2}\mathbb{E}_{\{u, \tilde{s}, \tilde{r}\}}\left[(\mathbf{u}_{m} \cdot \mathbf{u}_{m})(\mathbf{u}_{r} \cdot \mathbf{u}_{r})\right]=\sum_{r=1}^{q-2}\sum_{m=1}^{q-2}\left(n^2+2n \delta_{m, r}\right)=n^2 (q-2)^2+2n(q-2).
\end{aligned}
\end{equation}
\begin{equation}
\nonumber
\begin{aligned}
&\mathbb{E}_{\{u, \tilde{s}, \tilde{r}\}}\left[B^2\right]=\mathbb{E}_{\{u, \tilde{s}, \tilde{r}\}}\left[\left(\sum_{l=1}^{q-1}(\mathbf{w}_l \cdot \tilde{\mathbf{r}})^2\right)^2\right]=\sum_{l=1}^{q-1}\sum_{l'=1}^{q-1}\mathbb{E}_{\{u, \tilde{s}, \tilde{r}\}}\left[(\mathbf{w}_l \cdot \tilde{\mathbf{r}})^2(\mathbf{w}_{l'} \cdot \tilde{\mathbf{r}})^2\right]\\
&=\sum_{l=1}^{q-1}\sum_{l'=1}^{q-1}\sum_{r=1}^{q-2}\sum_{r'=1}^{q-2}\sum_{m=1}^{q-2}\sum_{m'=1}^{q-2}\alpha_{m, l}\alpha_{m', l}\alpha_{r, l'}\alpha_{r', l'}\mathbb{E}_{\{u, \tilde{s}, \tilde{r}\}}\left[(\mathbf{u}_{m} \cdot \tilde{\mathbf{r}})(\mathbf{u}_{m'} \cdot \tilde{\mathbf{r}})(\mathbf{u}_{r} \cdot \tilde{\mathbf{r}})(\mathbf{u}_{r'} \cdot \tilde{\mathbf{r}})\right]\\
&=\sum_{r=1}^{q-2}\sum_{r'=1}^{q-2}\sum_{m=1}^{q-2}\sum_{m'=1}^{q-2}\left(\sum_{l=1}^{q-1}\alpha_{m, l}\alpha_{m', l}\right)\left(\sum_{l'=1}^{q-1}\alpha_{r, l'}\alpha_{r', l'}\right)\mathbb{E}_{\{u, \tilde{s}, \tilde{r}\}}\left[(\mathbf{u}_{m} \cdot \tilde{\mathbf{r}})(\mathbf{u}_{m'} \cdot \tilde{\mathbf{r}})(\mathbf{u}_{r} \cdot \tilde{\mathbf{r}})(\mathbf{u}_{r'} \cdot \tilde{\mathbf{r}})\right]
\end{aligned}
\end{equation}
\begin{equation}
\begin{aligned}
&=\sum_{r=1}^{q-2}\sum_{r'=1}^{q-2}\sum_{m=1}^{q-2}\sum_{m'=1}^{q-2}\delta_{m, m'}\delta_{r, r'}\mathbb{E}_{\{u, \tilde{s}, \tilde{r}\}}\left[(\mathbf{u}_{m} \cdot \tilde{\mathbf{r}})(\mathbf{u}_{m'} \cdot \tilde{\mathbf{r}})(\mathbf{u}_{r} \cdot \tilde{\mathbf{r}})(\mathbf{u}_{r'} \cdot \tilde{\mathbf{r}})\right]\\
&=\sum_{r=1}^{q-2}\sum_{m=1}^{q-2}\mathbb{E}_{\{u, \tilde{s}, \tilde{r}\}}\left[(\mathbf{u}_{m} \cdot \tilde{\mathbf{r}})(\mathbf{u}_{m} \cdot \tilde{\mathbf{r}})(\mathbf{u}_{r} \cdot \tilde{\mathbf{r}})(\mathbf{u}_{r} \cdot \tilde{\mathbf{r}})\right]\\
&=\sum_{r=1}^{q-2}\sum_{m=1}^{q-2}\sum_{i=1}^{n}\sum_{j=1}^{n}\sum_{k=1}^{n}\sum_{z=1}^{n}\mathbb{E}_{u}\left[u_{m, i}u_{m, j}u_{r, k}u_{r, z}\right]\mathbb{E}_{\tilde{r}}\left[\tilde{r}_i\tilde{r}_j\tilde{r}_k\tilde{r}_z\right]\\
&=\sum_{r=1}^{q-2}\sum_{m=1}^{q-2}\sum_{i=1}^{n}\sum_{j=1}^{n}\sum_{k=1}^{n}\sum_{z=1}^{n}\mathbb{E}_{u}\left[u_{m, i}u_{m, j}u_{r, k}u_{r, z}\right]\left(\delta_{i, j}\delta_{k, z}+\delta_{i, k}\delta_{z, j}+\delta_{i, z}\delta_{k, j}\right)\\
&=\sum_{r=1}^{q-2}\sum_{m=1}^{q-2}\Big(\sum_{i=1}^{n}\sum_{k=1}^{n}\mathbb{E}_{u}\left[u_{m, i}u_{m, i}u_{r, k}u_{r, k}\right]+\sum_{i=1}^{n}\sum_{j=1}^{n}\mathbb{E}_{u}\left[u_{m, i}u_{m, j}u_{r, i}u_{r, j}\right]\\
&+\sum_{i=1}^{n}\sum_{j=1}^{n}\mathbb{E}_{u}\left[u_{m, i}u_{m, j}u_{r, j}u_{r, i}\right]\Big)=\sum_{r=1}^{q-2}\sum_{m=1}^{q-2}\Big(\sum_{i=1}^{n}\sum_{j=1}^{n}\mathbb{E}_{u}\left[u_{m, i}u_{m, i}u_{r, j}u_{r, j}\right]\\
&+\sum_{i=1}^{n}\sum_{j=1}^{n}\mathbb{E}_{u}\left[u_{m, i}u_{m, j}u_{r, i}u_{r, j}\right]+\sum_{i=1}^{n}\sum_{j=1}^{n}\mathbb{E}_{u}\left[u_{m, i}u_{m, j}u_{r, j}u_{r, i}\right]\Big)\\
&=\sum_{r=1}^{q-2}\sum_{m=1}^{q-2}\Big(\mathbb{E}_{u}\left[(\mathbf{u}_{m} \cdot \mathbf{u}_{m})(\mathbf{u}_{r} \cdot \mathbf{u}_{r})\right]+2\mathbb{E}_{u}\left[(\mathbf{u}_{m} \cdot \mathbf{u}_{r})(\mathbf{u}_{m} \cdot \mathbf{u}_{r})\right]\Big)\\
&=\sum_{r=1}^{q-2}\sum_{m=1}^{q-2}\Big(\left(n^2+2n\delta_{m, r}\right)+2\left((n^2+n)\delta_{m, r}+n\right)\Big)\\
&=(q-2)^2 n^2+2(q-2)n+2(q-2)n(n+1)+2(q-2)^2n=q(q-2)n(n+2).
\end{aligned}
\end{equation}
\begin{equation}
\begin{aligned}
&\mathbb{E}_{\{u, \tilde{s}, \tilde{r}\}}\left[AB\right]=\mathbb{E}_{\{u, \tilde{s}, \tilde{r}\}}\left[\left(\sum_{l=1}^{q-1}\mathbf{w}_l \cdot \mathbf{w}_l\right)\left(\sum_{l'=1}^{q-1}(\mathbf{w}_{l'} \cdot \tilde{\mathbf{r}})^2\right)\right]\\
&=\sum_{l=1}^{q-1}\sum_{l'=1}^{q-1}\sum_{r=1}^{q-2}\sum_{r'=1}^{q-2}\sum_{m=1}^{q-2}\sum_{m'=1}^{q-2}\alpha_{m, l}\alpha_{m', l}\alpha_{r, l'}\alpha_{r', l'}\mathbb{E}_{\{u, \tilde{s}, \tilde{r}\}}\left[(\mathbf{u}_{m} \cdot \mathbf{u}_{m'})(\mathbf{u}_{r} \cdot \tilde{\mathbf{r}})(\mathbf{u}_{r'} \cdot \tilde{\mathbf{r}})\right]\\
&=\sum_{r=1}^{q-2}\sum_{m=1}^{q-2}\mathbb{E}_{\{u, \tilde{s}, \tilde{r}\}}\left[(\mathbf{u}_{m} \cdot \mathbf{u}_{m})(\mathbf{u}_{r} \cdot \tilde{\mathbf{r}})^2\right]=\sum_{r=1}^{q-2}\sum_{m=1}^{q-2}\mathbb{E}_{u}\left[(\mathbf{u}_{m} \cdot \mathbf{u}_{m})\mathbb{E}_{\tilde{r}}\left[(\mathbf{u}_{r} \cdot \tilde{\mathbf{r}})^2\right]\right]\\
&=\sum_{r=1}^{q-2}\sum_{m=1}^{q-2}\mathbb{E}_{u}\left[(\mathbf{u}_{m} \cdot \mathbf{u}_{m})\sum_{i=1}^{n}\sum_{j=1}^{n}u_{r,i}u_{r,j}\mathbb{E}_{\tilde{r}}\left[\tilde{r}_i \tilde{r}_j\right]\right]=\sum_{r=1}^{q-2}\sum_{m=1}^{q-2}\mathbb{E}_{u}\left[(\mathbf{u}_{m} \cdot \mathbf{u}_{m})\sum_{i=1}^{n}\sum_{j=1}^{n}u_{r,i}u_{r,j}\delta_{i, j}\right]\\
&=\sum_{r=1}^{q-2}\sum_{m=1}^{q-2}\mathbb{E}_{u}\left[(\mathbf{u}_{m} \cdot \mathbf{u}_{m})\sum_{i=1}^{n}u_{r,i}^2\right]=\sum_{r=1}^{q-2}\sum_{m=1}^{q-2}\mathbb{E}_{u}\left[(\mathbf{u}_{m} \cdot \mathbf{u}_{m})(\mathbf{u}_{r} \cdot \mathbf{u}_{r})\right]=n^2 (q-2)^2+2n(q-2).
\end{aligned}
\end{equation}
\begin{equation}
\begin{aligned}
\mathbb{E}_{\{u, \tilde{s}, \tilde{r}\}}\left[aaA\right]=\mathbb{E}_{\{u, \tilde{s}, \tilde{r}\}}\left[\left(\tilde{\mathbf{r}} \cdot \tilde{\mathbf{r}}\right)\left(\sum_{l=1}^{q-1}\mathbf{w}_l \cdot \mathbf{w}_l\right)\right]=\mathbb{E}_{ \tilde{r}}\left[\tilde{\mathbf{r}} \cdot \tilde{\mathbf{r}}\right]\mathbb{E}_{u}\left[\sum_{l=1}^{q-1}\mathbf{w}_l \cdot \mathbf{w}_l\right]=n^2 (q-2).
\end{aligned}
\end{equation}
\begin{equation}
\begin{aligned}
\mathbb{E}_{\{u, \tilde{s}, \tilde{r}\}}\left[bbA\right]=\mathbb{E}_{\{u, \tilde{s}, \tilde{r}\}}\left[\left(\tilde{\mathbf{s}} \cdot \tilde{\mathbf{s}}\right)\left(\sum_{l=1}^{q-1}\mathbf{w}_l \cdot \mathbf{w}_l\right)\right]=\mathbb{E}_{ \tilde{s}}\left[\tilde{\mathbf{s}} \cdot \tilde{\mathbf{s}}\right]\mathbb{E}_{u}\left[\sum_{l=1}^{q-1}\mathbf{w}_l \cdot \mathbf{w}_l\right]=n^2 (q-2).
\end{aligned}
\end{equation}
\small
\begin{equation}
\begin{aligned}
&\mathbb{E}_{\{u, \tilde{s}, \tilde{r}\}}\left[aaB\right]=\mathbb{E}_{\{u, \tilde{s}, \tilde{r}\}}\left[\left(\tilde{\mathbf{r}} \cdot \tilde{\mathbf{r}}\right)\left(\sum_{l=1}^{q-1}(\mathbf{w}_l \cdot \tilde{\mathbf{r}})^2\right)\right]=\sum_{l=1}^{q-1}\sum_{m=1}^{q-2}\sum_{m'=1}^{q-2}\alpha_{m, l} \alpha_{m', l}\mathbb{E}_{\{u, \tilde{s}, \tilde{r}\}}\left[\left(\tilde{\mathbf{r}} \cdot \tilde{\mathbf{r}}\right)\left(\mathbf{u}_{m} \cdot \tilde{\mathbf{r}}\right)\left(\mathbf{u}_{m'} \cdot \tilde{\mathbf{r}}\right)\right]\\
&=\sum_{m=1}^{q-2}\sum_{m'=1}^{q-2}\delta_{m, m'}\mathbb{E}_{\{u, \tilde{s}, \tilde{r}\}}\left[\left(\tilde{\mathbf{r}} \cdot \tilde{\mathbf{r}}\right)\left(\mathbf{u}_{m} \cdot \tilde{\mathbf{r}}\right)\left(\mathbf{u}_{m'} \cdot \tilde{\mathbf{r}}\right)\right]=\sum_{m=1}^{q-2}\mathbb{E}_{\{u, \tilde{s}, \tilde{r}\}}\left[\left(\tilde{\mathbf{r}} \cdot \tilde{\mathbf{r}}\right)\left(\mathbf{u}_{m} \cdot \tilde{\mathbf{r}}\right)\left(\mathbf{u}_{m} \cdot \tilde{\mathbf{r}}\right)\right]\\
&=\sum_{m=1}^{q-2}\sum_{i=1}^{n}\sum_{j=1}^{n}\sum_{k=1}^{n}\mathbb{E}_{\{u, \tilde{s}, \tilde{r}\}}\left[u_{m, i}u_{m, j}\tilde{r}^2_k\tilde{r}_{i}\tilde{r}_{j}\right]=\sum_{m=1}^{q-2}\sum_{i=1}^{n}\sum_{j=1}^{n}\sum_{k=1}^{n}\mathbb{E}_{u}\left[u_{m, i}u_{m, j}\mathbb{E}_{\tilde{r}}\left[\tilde{r}^2_k\tilde{r}_{i}\tilde{r}_{j}\right]\right]\\
&=\sum_{m=1}^{q-2}\sum_{i=1}^{n}\sum_{j=1}^{n}\sum_{k=1}^{n}\mathbb{E}_{u}\left[u_{m, i}u_{m, j}(\delta_{i, j}+2 \delta_{i, k}\delta_{j, k})\right]=\sum_{m=1}^{q-2}\sum_{i=1}^{n}\sum_{k=1}^{n}\mathbb{E}_{u}\left[u^2_{m, i}\right]+2\sum_{m=1}^{q-2}\sum_{k=1}^{n}\mathbb{E}_{u}\left[u^2_{m, k}\right]\\
&=n^2(q-2)+2n(q-2)=(q-2)n(n+2).
\end{aligned}
\end{equation}
\normalsize
\small
\begin{equation}
\begin{aligned}
&\mathbb{E}_{\{u, \tilde{s}, \tilde{r}\}}\left[bbB\right]=\mathbb{E}_{\{u, \tilde{s}, \tilde{r}\}}\left[\left(\tilde{\mathbf{s}} \cdot \tilde{\mathbf{s}}\right)\left(\sum_{l=1}^{q-1}(\mathbf{w}_l \cdot \tilde{\mathbf{r}})^2\right)\right]=\sum_{l=1}^{q-1}\sum_{m=1}^{q-2}\sum_{m'=1}^{q-2}\alpha_{m, l} \alpha_{m', l}\mathbb{E}_{\{u, \tilde{s}, \tilde{r}\}}\left[\left(\tilde{\mathbf{s}} \cdot \tilde{\mathbf{s}}\right)\left(\mathbf{u}_{m} \cdot \tilde{\mathbf{r}}\right)\left(\mathbf{u}_{m'} \cdot \tilde{\mathbf{r}}\right)\right]\\
&=\sum_{m=1}^{q-2}\sum_{m'=1}^{q-2}\delta_{m, m'}\mathbb{E}_{\{u, \tilde{s}, \tilde{r}\}}\left[\left(\tilde{\mathbf{s}} \cdot \tilde{\mathbf{s}}\right)\left(\mathbf{u}_{m} \cdot \tilde{\mathbf{r}}\right)\left(\mathbf{u}_{m'} \cdot \tilde{\mathbf{r}}\right)\right]=\sum_{m=1}^{q-2}\mathbb{E}_{\tilde{s}}\left[\tilde{\mathbf{s}} \cdot \tilde{\mathbf{s}}\right]\mathbb{E}_{u, \tilde{r}}\left[\left(\mathbf{u}_{m} \cdot \tilde{\mathbf{r}}\right)\left(\mathbf{u}_{m} \cdot \tilde{\mathbf{r}}\right)\right]\\
&=n\sum_{m=1}^{q-2}\sum_{i=1}^n\sum_{j=1}^n\mathbb{E}_{u, \tilde{r}}\left[u_{m, i}u_{m, j}\tilde{r}_i \tilde{r}_j\right]=n\sum_{m=1}^{q-2}\sum_{i=1}^n\sum_{j=1}^n\mathbb{E}_{u}\left[u_{m, i}u_{m, j}\right]\mathbb{E}_{\tilde{r}}\left[\tilde{r}_i \tilde{r}_j\right]\\
&=n\sum_{m=1}^{q-2}\sum_{i=1}^n\sum_{j=1}^n\mathbb{E}_{u}\left[u_{m, i}u_{m, j}\right]\delta_{i, j}=n\sum_{m=1}^{q-2}\sum_{i=1}^n\mathbb{E}_{u}\left[u^2_{m, i}\right]=(q-2)n^2.
\end{aligned}
\end{equation}
\normalsize
\begin{equation}
\begin{aligned}
\mathbb{E}_{\{u, \tilde{s}, \tilde{r}\}}\left[ab^2A\right]=\mathbb{E}_{\{u, \tilde{s}, \tilde{r}\}}\left[\left(\tilde{\mathbf{r}}\cdot \tilde{\mathbf{s}}\right)^2\left(\sum_{l=1}^{q-1}\mathbf{w}_l \cdot \mathbf{w}_l\right)\right]=\mathbb{E}_{\{\tilde{s}, \tilde{r}\}}\left[\left(\tilde{\mathbf{r}}\cdot \tilde{\mathbf{s}}\right)^2\right]\mathbb{E}_{u}\left[\sum_{l=1}^{q-1}\mathbf{w}_l \cdot \mathbf{w}_l\right]=(q-2)n^2.
\end{aligned}
\end{equation}
\begin{equation}
\begin{aligned}
&\mathbb{E}_{\{u, \tilde{s}, \tilde{r}\}}\left[ab^2B\right]=\mathbb{E}_{\{u, \tilde{s}, \tilde{r}\}}\left[\left(\tilde{\mathbf{r}}\cdot \tilde{\mathbf{s}}\right)^2\left(\sum_{l=1}^{q-1}(\mathbf{w}_l \cdot \tilde{\mathbf{r}})^2\right)\right]\\
&=\sum_{l=1}^{q-1}\sum_{m=1}^{q-2}\sum_{m'=1}^{q-2}\sum_{i=1}^{n}\sum_{j=1}^{n}\sum_{k=1}^{n}\sum_{z=1}^{n}\alpha_{m, l}\alpha_{m', l}\mathbb{E}_{\{u, \tilde{s}, \tilde{r}\}}\left[\tilde{r}_i\tilde{s}_i\tilde{r}_j\tilde{s}_j \tilde{r}_k\tilde{r}_zu_{m, k}u_{m',z}\right]\\
&=\sum_{m=1}^{q-2}\sum_{i=1}^{n}\sum_{j=1}^{n}\sum_{k=1}^{n}\sum_{z=1}^{n}\mathbb{E}_{\{u, \tilde{s}, \tilde{r}\}}\left[\tilde{r}_i\tilde{s}_i\tilde{r}_j\tilde{s}_j \tilde{r}_k\tilde{r}_zu_{m, k}u_{m,z}\right]\\
&=\sum_{m=1}^{q-2}\sum_{i=1}^{n}\sum_{j=1}^{n}\sum_{k=1}^{n}\sum_{z=1}^{n}\mathbb{E}_{ \tilde{r}}\left[\tilde{r}_i\tilde{r}_j \tilde{r}_k\tilde{r}_z\right]\mathbb{E}_{\tilde{s}}\left[\tilde{s}_i \tilde{s}_j\right]\mathbb{E}_{u}\left[u_{m, k}u_{m,z}\right]\\
&=\sum_{m=1}^{q-2}\sum_{i=1}^{n}\sum_{j=1}^{n}\sum_{k=1}^{n}\sum_{z=1}^{n}\mathbb{E}_{ \tilde{r}}\left[\tilde{r}_i\tilde{r}_j \tilde{r}_k\tilde{r}_z\right]\delta_{i, j}\mathbb{E}_{u}\left[u_{m, k}u_{m,z}\right]\\
&=\sum_{m=1}^{q-2}\sum_{i=1}^{n}\sum_{k=1}^{n}\sum_{z=1}^{n}\mathbb{E}_{ \tilde{r}}\left[\tilde{r}^2_i \tilde{r}_k\tilde{r}_z\right]\mathbb{E}_{u}\left[u_{m, k}u_{m,z}\right]=\sum_{m=1}^{q-2}\sum_{i=1}^{n}\sum_{k=1}^{n}\sum_{z=1}^{n}(\delta_{k, z}+2\delta_{i, k}\delta_{i, z})\mathbb{E}_{u}\left[u_{m, k}u_{m,z}\right]\\
&=\sum_{m=1}^{q-2}\sum_{i=1}^{n}\sum_{k=1}^{n}\mathbb{E}_{u}\left[u^2_{m, k}\right]+2\sum_{m=1}^{q-2}\sum_{i=1}^{n}\mathbb{E}_{u}\left[u^2_{m, i}\right]=(q-2)n(n+2).
\end{aligned}
\end{equation}
\begin{equation}
\begin{aligned}
&\mathbb{E}_{\{u, \tilde{s}, \tilde{r}\}}\left[F\right]=\mathbb{E}_{u}\left[\sum_{l=1}^{q-1}(\mathbf{w}_l \cdot \mathbf{w}_l)^2\right]=\sum_{l=1}^{q-1}\sum_{m=1}^{q-2}\sum_{m'=1}^{q-2}\sum_{r=1}^{q-2}\sum_{r'=1}^{q-2}\alpha_{m, l}\alpha_{m', l}\alpha_{r, l}\alpha_{r', l}\mathbb{E}_{u}\left[(\mathbf{u}_m \cdot \mathbf{u}_{m'})(\mathbf{u}_r \cdot \mathbf{u}_{r'})\right]\\
&=\sum_{l=1}^{q-1}\sum_{m=1}^{q-2}\sum_{m'=1}^{q-2}\sum_{r=1}^{q-2}\sum_{r'=1}^{q-2}\alpha_{m, l}\alpha_{m', l}\alpha_{r, l}\alpha_{r', l}\left(n^2 \delta_{m, m'}\delta_{r, r'}+n \delta_{m, r}\delta_{m', r'}+n \delta_{m, r'}\delta_{m', r}\right)\\
&=(n^2+2n)\sum_{l=1}^{q-1}\sum_{m=1}^{q-2}\sum_{r=1}^{q-2}\alpha^2_{m, l}\alpha^2_{r, l}.
\end{aligned}
\end{equation}
\begin{equation}
\begin{aligned}
&\mathbb{E}_{\{u, \tilde{s}, \tilde{r}\}}\left[H\right]=\mathbb{E}_{\{u, \tilde{r}\}}\left[\sum_{l=1}^{q-1}(\mathbf{w}_l \cdot \mathbf{w}_l)(\mathbf{w}_l \cdot \tilde{\mathbf{r}})^2\right]\\
&=\sum_{l=1}^{q-1}\sum_{m=1}^{q-2}\sum_{m'=1}^{q-2}\sum_{r=1}^{q-2}\sum_{r'=1}^{q-2}\alpha_{m, l}\alpha_{m', l}\alpha_{r, l}\alpha_{r', l}\mathbb{E}_{\{u, \tilde{r}\}}\left[(\mathbf{u}_m \cdot \mathbf{u}_{m'})(\mathbf{u}_r \cdot \tilde{\mathbf{r}})(\mathbf{u}_{r'} \cdot \tilde{\mathbf{r}})\right]\\
&=\sum_{l=1}^{q-1}\sum_{m=1}^{q-2}\sum_{m'=1}^{q-2}\sum_{r=1}^{q-2}\sum_{r'=1}^{q-2}\sum_{i=1}^{n}\sum_{j=1}^{n}\sum_{k=1}^{n}\alpha_{m, l}\alpha_{m', l}\alpha_{r, l}\alpha_{r', l}\mathbb{E}_{\{u, \tilde{r}\}}\left[u_{m, i}u_{m', i}u_{r, j}u_{r', k}\tilde{r}_{j}\tilde{r}_{k}\right]\\
&=\sum_{l=1}^{q-1}\sum_{m=1}^{q-2}\sum_{m'=1}^{q-2}\sum_{r=1}^{q-2}\sum_{r'=1}^{q-2}\sum_{i=1}^{n}\sum_{j=1}^{n}\sum_{k=1}^{n}\alpha_{m, l}\alpha_{m', l}\alpha_{r, l}\alpha_{r', l}\mathbb{E}_{u}\left[u_{m, i}u_{m', i}u_{r, j}u_{r', k}\right]\mathbb{E}_{\tilde{r}}\left[\tilde{r}_{j}\tilde{r}_{k}\right]\\
&=\sum_{l=1}^{q-1}\sum_{m=1}^{q-2}\sum_{m'=1}^{q-2}\sum_{r=1}^{q-2}\sum_{r'=1}^{q-2}\sum_{i=1}^{n}\sum_{j=1}^{n}\sum_{k=1}^{n}\alpha_{m, l}\alpha_{m', l}\alpha_{r, l}\alpha_{r', l}\mathbb{E}_{u}\left[u_{m, i}u_{m', i}u_{r, j}u_{r', k}\right]\delta_{j, k}\\
&=\sum_{l=1}^{q-1}\sum_{m=1}^{q-2}\sum_{m'=1}^{q-2}\sum_{r=1}^{q-2}\sum_{r'=1}^{q-2}\sum_{i=1}^{n}\sum_{j=1}^{n}\alpha_{m, l}\alpha_{m', l}\alpha_{r, l}\alpha_{r', l}\mathbb{E}_{u}\left[u_{m, i}u_{m', i}u_{r, j}u_{r', j}\right]\\
&=\sum_{l=1}^{q-1}\sum_{m=1}^{q-2}\sum_{m'=1}^{q-2}\sum_{r=1}^{q-2}\sum_{r'=1}^{q-2}\alpha_{m, l}\alpha_{m', l}\alpha_{r, l}\alpha_{r', l}\mathbb{E}_{u}\left[(\mathbf{u}_m \cdot \mathbf{u}_{m'})(\mathbf{u}_r \cdot \mathbf{u}_{r'})\right]\\
&=\sum_{l=1}^{q-1}\sum_{m=1}^{q-2}\sum_{m'=1}^{q-2}\sum_{r=1}^{q-2}\sum_{r'=1}^{q-2}\alpha_{m, l}\alpha_{m', l}\alpha_{r, l}\alpha_{r', l}\left(n^2 \delta_{m, m'}\delta_{r, r'}+n \delta_{m, r}\delta_{m', r'}+n \delta_{m, r'}\delta_{m', r}\right)\\
&=n(n+2)\sum_{l=1}^{q-1}\sum_{m=1}^{q-2}\sum_{r=1}^{q-2}\alpha^2_{m, l}\alpha^2_{r, l}.
\end{aligned}
\end{equation}
\small
\begin{equation}
\nonumber
\begin{aligned}
&\mathbb{E}_{\{u, \tilde{s}, \tilde{r}\}}\left[J\right]=\mathbb{E}_{\{u, \tilde{r}\}}\left[\sum_{l=1}^{q-1}(\mathbf{w}_l \cdot \tilde{\mathbf{r}})^4\right]\\
&=\sum_{l=1}^{q-1}\sum_{m=1}^{q-2}\sum_{m'=1}^{q-2}\sum_{r=1}^{q-2}\sum_{r'=1}^{q-2}\alpha_{m, l}\alpha_{m', l}\alpha_{r, l}\alpha_{r', l}\mathbb{E}_{\{u, \tilde{r}\}}\left[(\mathbf{u}_m \cdot \tilde{\mathbf{r}})(\mathbf{u}_{m'} \cdot \tilde{\mathbf{r}})(\mathbf{u}_r \cdot \tilde{\mathbf{r}})(\mathbf{u}_{r'} \cdot \tilde{\mathbf{r}})\right]\\
&=\sum_{l=1}^{q-1}\sum_{m=1}^{q-2}\sum_{m'=1}^{q-2}\sum_{r=1}^{q-2}\sum_{r'=1}^{q-2}\sum_{i=1}^{n}\sum_{j=1}^{n}\sum_{k=1}^{n}\sum_{z=1}^{n}\alpha_{m, l}\alpha_{m', l}\alpha_{r, l}\alpha_{r', l}\mathbb{E}_{\{u, \tilde{r}\}}\left[u_{m, i}u_{m', j}u_{r, k}u_{r', z}\tilde{r}_{i}\tilde{r}_{j}\tilde{r}_{k}\tilde{r}_{z}\right]\\
&=\sum_{l=1}^{q-1}\sum_{m=1}^{q-2}\sum_{m'=1}^{q-2}\sum_{r=1}^{q-2}\sum_{r'=1}^{q-2}\sum_{i=1}^{n}\sum_{j=1}^{n}\sum_{k=1}^{n}\sum_{z=1}^{n}\alpha_{m, l}\alpha_{m', l}\alpha_{r, l}\alpha_{r', l}\mathbb{E}_{u}\left[u_{m, i}u_{m', j}u_{r, k}u_{r', z}\right]\mathbb{E}_{\tilde{r}}\left[\tilde{r}_{i}\tilde{r}_{j}\tilde{r}_{k}\tilde{r}_{z}\right]\\
&=\sum_{l=1}^{q-1}\sum_{m=1}^{q-2}\sum_{m'=1}^{q-2}\sum_{r=1}^{q-2}\sum_{r'=1}^{q-2}\sum_{i=1}^{n}\sum_{j=1}^{n}\sum_{k=1}^{n}\sum_{z=1}^{n}\alpha_{m, l}\alpha_{m', l}\alpha_{r, l}\alpha_{r', l}\mathbb{E}_{u}\left[u_{m, i}u_{m', j}u_{r, k}u_{r', z}\right]\Big(\delta_{i, j}\delta_{k, z}+\delta_{i, k}\delta_{j, z}+\delta_{i, z}\delta_{j, k}\Big)\\
&=\sum_{l=1}^{q-1}\sum_{m=1}^{q-2}\sum_{m'=1}^{q-2}\sum_{r=1}^{q-2}\sum_{r'=1}^{q-2}\alpha_{m, l}\alpha_{m', l}\alpha_{r, l}\alpha_{r', l}\Big(\sum_{i=1}^{n}\sum_{k=1}^{n}\mathbb{E}_{u}\left[u_{m, i}u_{m', i}u_{r, k}u_{r', k}\right]+\sum_{i=1}^{n}\sum_{j=1}^{n}\mathbb{E}_{u}\left[u_{m, i}u_{m', j}u_{r, i}u_{r', j}\right]
\end{aligned}
\end{equation}
\normalsize
\begin{equation}
\begin{aligned}
&+\sum_{i=1}^{n}\sum_{j=1}^{n}\mathbb{E}_{u}\left[u_{m, i}u_{m', j}u_{r, j}u_{r', i}\right]\Big)=\sum_{l=1}^{q-1}\sum_{m=1}^{q-2}\sum_{m'=1}^{q-2}\sum_{r=1}^{q-2}\sum_{r'=1}^{q-2}\alpha_{m, l}\alpha_{m', l}\alpha_{r, l}\alpha_{r', l}\Big(
\mathbb{E}_{u}\left[(\mathbf{u}_{m} \cdot \mathbf{u}_{m'})(\mathbf{u}_{r} \cdot \mathbf{u}_{r'})\right]\\
&+\mathbb{E}_{u}\left[(\mathbf{u}_{m} \cdot \mathbf{u}_{r})(\mathbf{u}_{m'} \cdot \mathbf{u}_{r'})\right]+\mathbb{E}_{u}\left[(\mathbf{u}_{m} \cdot \mathbf{u}_{r'})(\mathbf{u}_{m'} \cdot \mathbf{u}_{r})\right]\Big)\\
&=3\sum_{l=1}^{q-1}\sum_{m=1}^{q-2}\sum_{m'=1}^{q-2}\sum_{r=1}^{q-2}\sum_{r'=1}^{q-2}\alpha_{m, l}\alpha_{m', l}\alpha_{r, l}\alpha_{r', l}
\mathbb{E}_{u}\left[(\mathbf{u}_{m} \cdot \mathbf{u}_{m'})(\mathbf{u}_{r} \cdot \mathbf{u}_{r'})\right]\\
&=3\sum_{l=1}^{q-1}\sum_{m=1}^{q-2}\sum_{m'=1}^{q-2}\sum_{r=1}^{q-2}\sum_{r'=1}^{q-2}\alpha_{m, l}\alpha_{m', l}\alpha_{r, l}\alpha_{r', l}
\left(n^2\delta_{m,m'}\delta_{r,r'}+n\delta_{m,r}\delta_{m',r'}+n\delta_{m,r'}\delta_{m',r}\right)\\
&=3n(n+2)\sum_{l=1}^{q-1}\sum_{m=1}^{q-2}\sum_{r=1}^{q-2}\alpha^2_{m, l}\alpha^2_{r, l}.
\end{aligned}
\end{equation}
\normalsize
Direct substitution yields
\begin{align}
\mathbb{E}_{u,\tilde{s},\tilde{r}}[c_1]=\dfrac{n(2q-1)}{2q(q-1)},\quad \mathbb{E}_{u,\tilde{s},\tilde{r}}[c_2]=\dfrac{n(n+q)(q-1)}{4q^2},
\end{align}
and so we have
\begin{equation}
\label{sseries}
S=-\dfrac{n(2q-1)}{2q(q-1)}\mu-\dfrac{n(n+q)(q-1)}{4q^2}\mu^2+\mathcal{O}\left(\mu^{3}\right).
\end{equation}
From Equation \eqref{sseries} and Theorem \ref{entrocaltheo}, we have
\begin{equation}
\begin{aligned}
&h(\vecX|\hat\vecW)=n h_{\sigma}+\dfrac{n}{2}\left(1-\dfrac{1}{q}\right)\mu-\dfrac{n(n+q)(q-1)}{4q^2}\mu^2+\mathcal{O}\left(\mu^{3}\right),\\
&=n h_{\sigma}+\dfrac{n}{2}\left(1-\dfrac{1}{q}\right)\mu-\dfrac{n}{2}\left(1-\dfrac{1}{q}\right)\dfrac{nq^{-1}+1}{2q}\mu^2+\mathcal{O}\left(\mu^{3}\right).
\end{aligned}
\end{equation}
\end{document}